\documentclass{article} 
\usepackage{iclr2027_conference,times}

\usepackage{amsmath,amsfonts,bm}

\def\eqref#1{equation~\ref{#1}}

\def\1{\bm{1}}

\DeclareMathAlphabet{\mathsfit}{\encodingdefault}{\sfdefault}{m}{sl}
\SetMathAlphabet{\mathsfit}{bold}{\encodingdefault}{\sfdefault}{bx}{n}

\usepackage{hyperref}
\usepackage{url}
\usepackage{graphicx}
\usepackage{subcaption}
\usepackage{booktabs}
\usepackage{siunitx}
\usepackage{amssymb}
\usepackage{amsfonts}
\usepackage{amsmath}
\usepackage{amsthm}

\usepackage{algorithm}
\usepackage{algpseudocode}
\usepackage{tabularx}
\theoremstyle{plain}
\newtheorem{theorem}{Theorem}[section]

\newtheorem{proposition}[theorem]{Proposition}

\theoremstyle{definition}

\newtheorem{assumption}[theorem]{Assumption}

\newtheorem{problem}[theorem]{Problem}

\theoremstyle{remark}

\DeclareMathOperator*{\softmin}{softmin}
\usepackage{adjustbox}

\usepackage{wrapfig}

\title{SAFTER: Safety-Aware Filtering Through Explicit Recoverability for Robust Reinforcement Learning under Distribution Shift}

\author{Benedictus C. G. Cinun \\
Department of Electrical and Computer Engineering\\
University of Houston\\
Houston, TX 77004, USA \\
\texttt{bgerodacinun@uh.edu} \\
\And
Bin Hu\\
Department of Engineering Technology \\
University of Houston \\
Houston, TX 77004, USA \\
\texttt{bhu11@central.uh.edu} \\
}

\iclrfinalcopy 
\begin{document}

\maketitle

\begin{abstract}
Learning safe policies that generalize to unseen environments is difficult under bounded actuation, where an unsafe decision may become unrecoverable.
Conventional control barrier function (CBF) filters enforce only instantaneous safety and ignore future recoverability.
We propose SAFTER (Safety-Aware Filtering Through Explicit Recoverability), which embeds a differentiable, closed-form backup-CBF filter into Soft Actor-Critic training.
SAFTER unifies current safety, finite-horizon backup safety, terminal recoverability, and input admissibility into a single differentiable certificate.
Its CBF condition yields an exact minimum-intervention projection without iterative online optimization.
In quadruped navigation, SAFTER reaches $90\%$ training success with $30.6\%$ fewer transitions than Vanilla SAC and $34.1$--$35.1\%$ fewer than current-state CBF-filtered SAC baselines.
Across unseen simulated environments, SAFTER improves success by up to $42.4$ and $61.4$ percentage points over the state-of-the-art methods CBF--RL and SEA--Nav, respectively, and achieves $73.3\%$ and $76.7\%$ success on two real-world navigation maps.
Ablations identify accurate differentiation through the safety filter as a key contributor to robustness, in which replacing the projection Jacobian with the identity approximation reduces success rate in a High-Shift environment from $61.4\%$ to $6.2\%$.
These results show that recoverability-aware filtering serves as both a runtime safety mechanism and a learning mechanism for sample-efficient, robust policies.
\end{abstract}

\section{Introduction}
Reinforcement learning (RL) can produce adaptive navigation policies directly from interaction and local sensory observations, but trial-and-error exploration can generate unsafe actions during both training and deployment \citep{brunke2022safe,guerrier2024learning,xu2022benchmarking}. Reward penalties can discourage collisions or constraint violations, yet they provide only an indirect learning signal and do not ensure that each executed action satisfies a prescribed safety condition \citep{krasowski2022provably}. For safety-critical robotic navigation, a useful learning architecture should therefore protect the system during interaction while still allowing the policy to improve from the intervention it receives.

Control barrier functions (CBFs) provide a principled way to enforce state constraints by minimally modifying a nominal control input so that a safe set remains forward invariant \citep{ames2019control}. This reactive view, however, is incomplete under bounded actuation. A system can satisfy all instantaneous geometric constraints and still enter a state from which the available control authority is no longer sufficient to prevent a future violation. Backup CBFs (bCBFs) address this distinction by asking a stronger question: \emph{if the nominal behavior becomes unsafe, does a prescribed recovery maneuver remain feasible over a finite horizon and reach a terminal safe set?} \citep{gurriet2018online,chen2021backup}. This notion of future recoverability is particularly relevant to learning-based navigation because the policy should ideally avoid actions that are safe only momentarily but eliminate subsequent braking or avoidance authority.

The challenge is to make this predictive notion of safety practical inside RL training. A sampled bCBF introduces many simultaneous conditions: current-state safety, safety at multiple future backup states, terminal recoverability, and bounded-input constraints. Direct enforcement therefore leads to a multi-constraint online optimization that must be solved at every environment interaction. Existing work has separately developed smooth CBF composition, closed-form CBF filters, backup-based safe sets, and differentiable safety layers \citep{rabiee2025soft,wang2025multi,van2026safety,yun2026provably}. What remains important for learning is a formulation that preserves finite-horizon backup recoverability, handles bounded commands, collapses the resulting constraint collection to an analytical minimum-intervention map, and exposes the exact filter Jacobian to the task policy during training.

We introduce \textbf{SAFTER}, a recoverability-aware safety-filtered RL framework built around this objective. Figure~\ref{fig:architecture} summarizes the central idea. The actor proposes a task-oriented nominal action. SAFTER predicts a prescribed backup maneuver, composes current, future, terminal, and input-admissibility conditions into a single differentiable certificate, and projects the nominal command through an exact closed-form safety map. During training, the filtered action is the action actually executed and evaluated by the critic, while the actor differentiates through the analytical safety correction. During deployment, the actor is frozen and the same filter remains active. Thus, the safety layer is used not only to shield the robot but also to shape what the policy learns.

Our contributions are threefold. \textbf{First}, we construct a recoverability-aware composite certificate that combines current-state HOCBF conditions, sampled backup-horizon safety, terminal recoverability, and bounded-input admissibility after control-input augmentation. The construction reduces the sampled predictive-safety conditions to one affine CBF constraint whose exact minimum-intervention solution is available in closed form. \textbf{Second}, we embed this analytical filter directly into SAC so that the critic learns from executed filtered actions and the actor receives gradients through the exact recoverability-aware filter mapping. \textbf{Third}, controlled quadruped experiments separate learning efficiency, policy internalization, distribution-shift robustness, and filter-gradient effects. The results show faster transition-based convergence in a matched moderate-speed study and substantially improved navigation success in unseen, geometrically shifted environments; ablations identify the analytical filter Jacobian as an important learning mechanism.

\section{Related Work}
\textbf{Safety-filtered reinforcement learning.}
Safety layers and shielding mechanisms separate task learning from explicit safety enforcement, while constrained and safety-critic formulations provide alternative mechanisms for safe policy learning
\citep{chow2021safe,hsu2023isaacs,yang2023safety}.
Reward shaping can bias policies away from unsafe behavior \citep{ranjan2024barrier}, while explicit locomotion constraints can reduce reward engineering \citep{kim2024not}.
Training-time safety filtering has also been shown to help the policy account for filter interventions and improve sample efficiency \citep{bejarano2024safety,van2024safe}.
Reachability-based approaches can switch between task and recovery policies \citep{he2024agile}, while observation-conditioned reachability filtering has been developed for quadruped navigation in previously unknown environments \citep{lin2025one}.
In legged navigation, SEA--Nav learns safe and agile LiDAR-based locomotion \citep{chen2026sea}, and CBF--RL combines closed-form CBF filtering with barrier-inspired reward shaping so that the task policy can internalize safety behavior \citep{yang2025cbf}.
These approaches motivate training-time safety integration but do not directly provide sampled finite-horizon backup recoverability through the task-policy action map.

\textbf{Backup CBFs and predictive safety.}
Backup CBFs characterize recoverability through the flow of a prescribed backup controller, beginning with active-set invariance formulations \citep{gurriet2018online} and later extensions using time-varying backups \citep{singletary2022safe}.
RLBUS learns a neural backup policy to enlarge the finite-time safe backward image \citep{rabiee2023safe}; PS2-RL couples task-policy learning with a differentiable backup-CBF projection and considers both learned and analytical backup constructions \citep{yun2026provably}.
Backup-CBF filtering has also been developed for bounded-input and higher-order systems \citep{van2026safety,lin2026safety}.
SAFTER instead combines an augmented composite backup certificate with an explicit projection in surrogate-control space and differentiates through this map during off-policy task-policy learning.

\textbf{Composite and explicit safety filters.}
LSE composition has been used to combine multiple current-state CBF constraints and obtain a closed-form safety filter for end-to-end RL \citep{wang2025multi}.
CBF-based quadruped navigation has also used smooth LiDAR-derived safety functions with online QP filtering \citep{unlu2024control}.
For backup barriers, smooth minimum/maximum constructions aggregate finite-horizon conditions while retaining optimization-based control synthesis \citep{rabiee2025soft}.
\citet{van2026safety} derive a closed-form bounded-input bCBF controller by interpolating between nominal and backup controls, restricting the correction to that interpolation direction.
SAFTER instead constructs a single composite affine constraint in an augmented surrogate-input space and applies the exact minimum-Euclidean correction in the full surrogate control space.
The differentiable mapping is then used directly in the SAC actor update.


\section{Recoverability-Aware Safe RL: Problem Formulation}
\label{sec:preliminaries}

\textbf{Safety-filtered Soft Actor-Critic RL.}
\label{sec:safe_rl}
We consider continuous-action reinforcement learning with environment state
$s_t$, policy observation $o_t$, and physical command $u_t$. The task policy
is trained using Soft Actor-Critic (SAC) \citep{haarnoja2018soft}. At policy
step $t$, the actor proposes a nominal command
$u_t^{\mathrm{nom}}\sim\pi_\phi(\cdot\mid o_t)$, which is transformed before
execution by $u_t^{\mathrm{safe}} = \mathcal F_\eta(\xi_t,u_t^{\mathrm{nom}}),$ where $\xi_t$ contains the state and sensing information required by the filter and $\eta$ denotes its fixed design parameters. Because the executed
command generates the environment transition, it is the action stored in
replay and evaluated by the SAC critics.
During the actor update, SAFTER also differentiates through $\mathcal F_\eta$, allowing the nominal policy to adapt
to the safety intervention it experiences.

The safety filter doesn't need the complete RL environment.
We use $x\in\mathbb R^n$ for the continuous state required by the safety model and $u\in\mathcal U$ for its physical command, with local dynamics $\dot x=f(x)+g(x)u$. Thus, $s_t$ is the RL environment state, $o_t$ is the actor observation, and $x$ is the reduced state used to evaluate safety.

\textbf{From current-state safety to backup recoverability.}
\label{sec:cbf_bcbf}
Let $\mathcal C_S=\{x:h_i(x)\geq0,\ i=1,\ldots,N_h\}$ denote the prescribed
safe set. For a differentiable relative-degree-one barrier $h$, a standard
minimum-intervention CBF filter solves
\begin{equation}
\begin{aligned}
    u^{\mathrm{safe}}
    ={}&\arg\min_{u\in\mathcal U}
        \frac12\|u-u^{\mathrm{nom}}\|_2^2\\
    \mathrm{s.t.}\quad&
        L_fh(x)+L_gh(x)u+\alpha(h(x))\geq0,
\end{aligned}
\label{eq:standard_cbf_qp}
\end{equation}
where $L_fh(x)=\nabla h(x)^\top f(x)$ and
$L_gh(x)=\nabla h(x)^\top g(x)$ are the Lie derivatives of $h$, and
$\alpha$ regulates approach to the safety boundary. Under the usual
regularity and feasibility assumptions, the CBF constraint renders the
corresponding safe set forward invariant \citep{ames2019control}; the
objective leaves the nominal command unchanged whenever it already satisfies the
input and CBF constraints.

A conventional CBF constrains the instantaneous safety-margin dynamics, but
with bounded actuation this need not preserve future recoverability. Let
$u_{\mathrm b}(x)$ be a prescribed backup controller,
$\varphi_{\mathrm b}(\tau;x)$ its flow, and
$\mathcal C_B^x=\{x:h_{B,j}(x)\geq0,\ j=1,\ldots,N_B\}\subseteq\mathcal C_S$
a terminal set invariant under $u_{\mathrm b}$~\citep{gurriet2018online,chen2021backup}. Finite-horizon recoverability
requires
\[
    h_i\!\left(\varphi_{\mathrm b}(\tau;x)\right)\geq0,
    \quad \forall \tau\in[0,T],\ \forall i,
    \qquad
    h_{B,j}\!\left(\varphi_{\mathrm b}(T;x)\right)\geq0,
    \quad \forall j.
\]
The backup trajectory must therefore remain safe while reaching a terminal
state from which safety can subsequently be maintained. Evaluating the
trajectory at $0<\tau_1<\cdots<\tau_K=T$ produces $N_hK$ future safety
conditions in addition to current and terminal conditions. Direct enforcement
therefore leads to the sampled multi-constraint bCBF--QP described in
Appendix~\ref{app:sampled_bcbf_qp}: recoverability provides predictive safety
information unavailable to an instantaneous CBF, but substantially enlarges
the online safety problem.

\textbf{Quadruped-navigation example.}
In our experiments, $s_t$ contains the full MuJoCo state, while the SAC actor
observes local LiDAR, goal information, measured body velocity, and command
history. The safety model uses the scan-fixed planar pose
$x=[p_x,p_y,\theta]^\top$ and high-level command
$u=[v_x,v_y,\omega_z]^\top$. Each valid LiDAR return defines an obstacle
point and barrier $h_i(x)$. A command can satisfy all current clearance
constraints yet leave too little distance to stop before collision. The
prescribed backup therefore drives the commanded motion toward rest and tests
whether this stopping maneuver remains collision-free over the horizon. With
$N_h=120$ rays and $K=20$ positive-time samples, the filter evaluates $2400$
future geometric barriers before terminal and input-admissibility conditions
are included. Appendix~\ref{app:quadruped_safety} gives the full construction.

\label{sec:problem_formulation}
\begin{problem}[Recoverability-aware differentiable safe RL]
\label{prob:safer}
Given $\pi_\phi$, the safety model $(f,g)$, safe set $\mathcal C_S$, bounded
input set $\mathcal U$, backup controller $u_{\mathrm b}$, and terminal set
$\mathcal C_B^x$, construct $\mathcal F_\eta$ such that the executed command:
(i) satisfies the encoded current, sampled backup, terminal, and input
conditions whenever feasible; (ii) minimally modifies the nominal task
command; and (iii) can be evaluated efficiently and differentiated during
SAC training.
\end{problem}

Direct sampled bCBF enforcement satisfies the predictive requirement but
retains a multi-constraint optimization inside every RL interaction.
Section~\ref{sec:proposed_method} develops SAFTER to remove this bottleneck by
combining bounded-command augmentation, finite-horizon backup prediction, a
single composite certificate, and an analytical minimum-intervention
projection whose local Jacobian can be used directly in policy optimization.

\section{Recoverability-Aware Differentiable Safety Filtering}
\label{sec:proposed_method}

SAFTER realizes Problem~\ref{prob:safer} through the analytical pipeline in
Figure~\ref{fig:architecture}. It augments the bounded physical command with
first-order dynamics, evaluates finite-horizon backup recoverability in the
augmented model, and combines the resulting current, future, terminal, and
input-admissibility margins into one differentiable certificate. The CBF
condition of this certificate is a single affine constraint in a surrogate
input, yielding an exact minimum-intervention projection and an exact local
Jacobian wherever the filter is differentiable.

The central technical contribution is this reduction from sampled
multi-constraint backup-CBF enforcement to a single analytical projection
after the constituent safety margins have been evaluated. Thus, the number
of backup samples and primitive barriers no longer determines the number of
constraints in the online optimization, while the same recoverability-aware
action map can be differentiated through during SAC training.

\begin{figure}[t]
    \centering
    \includegraphics[width=\linewidth]{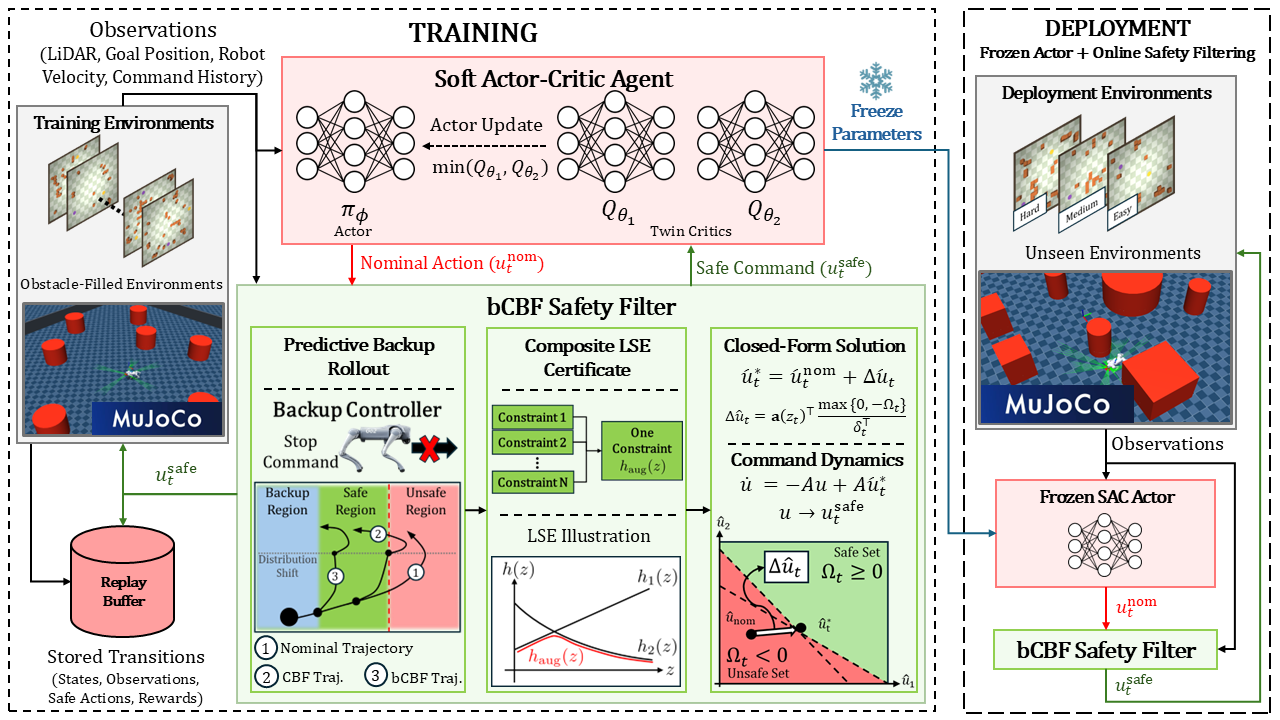}
    \caption{Overview of SAFTER. A task policy proposes a nominal command; a
    prescribed backup rollout evaluates finite-horizon recoverability;
    current, future, terminal, and input-admissibility conditions are combined
    into one differentiable certificate; and an analytical minimum-intervention
    projection determines the safety correction before execution. During
    training, the critic evaluates the executed action and the actor
    differentiates through the same safety map.}
    \label{fig:architecture}
\end{figure}

\textbf{Bounded-command augmentation and current safety.}
\label{sec:augmentation_lifting}
Section~\ref{sec:preliminaries} uses the safety model
$\dot x=f(x)+g(x)u$ with bounded physical command $u\in\mathcal U$.
Following the control-dynamics construction of \citet{rabiee2024closed}, we
introduce a surrogate input $\hat u\in\mathbb R^m$ through
$\dot u=-Au+A\hat u$, where
$A=\operatorname{diag}(a_1,\ldots,a_m)\succ0$. Treating the applied command as
part of the state gives $z=[x^\top,u^\top]^\top$ and
\begin{equation}
    \dot z=\hat f(z)+\hat g(z)\hat u,
    \qquad
    \hat f(z)=
    \begin{bmatrix}
        f(x)+g(x)u\\[-1mm]
        -Au
    \end{bmatrix},
    \qquad
    \hat g(z)=
    \begin{bmatrix}
        0_{n\times m}\\[-1mm]
        A
    \end{bmatrix}.
    \label{eq:augmented_dynamics}
\end{equation}
For the component-wise bounds $u_{\min}\leq u\leq u_{\max}$, define
$I_d^+(z):=u_{\max,d}-u_d$ and
$I_d^-(z):=u_d-u_{\min,d}$, $d=1,\ldots,m$. Their nonnegativity represents
the physical command bounds as state constraints. Because $h_i$ depends on $x$ but not directly on the command state $u$,
$L_{\hat g}h_i=0$. We therefore use the lifted current-safety margin
\[
    \psi_i(z)
    :=L_fh_i(x)+L_gh_i(x)u+k_0h_i(x),
    \qquad
    L_{\hat g}\psi_i(z)=L_gh_i(x)A,
\]
with $k_0>0$. Hence, $\psi_i$ has relative degree one with respect to
$\hat u$ wherever $L_gh_i(x)A\neq0$ and can be combined with the predictive
and input-admissibility conditions below. Appendix~\ref{app:hocbf_relative_degree}
gives the relative-degree calculation, while Appendix~\ref{app:finite_rate_realization} describes the finite-rate command
update and the post-projection enforcement of the physical command limits.

\textbf{Predictive backup recoverability.}
\label{sec:predictive_backup}
Because SAFTER acts through $\hat u$, the backup maneuver $u_{\mathrm b}$ from
Section~\ref{sec:preliminaries} is represented in the augmented dynamics by a
prescribed surrogate backup input $\hat u_{\mathrm b}(z)$. Let
$\varphi_{\mathrm b}^{z}(\tau;z)$ denote the corresponding flow of
$\dot z=\hat f(z)+\hat g(z)\hat u_{\mathrm b}(z)$. With
$E_x=[I_n\;0_{n\times m}]$, the sampled future margins are
$
    h_i^{\mathrm b}(z,\tau_k)
    :=h_i\!\left(E_x\varphi_{\mathrm b}^{z}(\tau_k;z)\right).
$
After augmentation, terminal recoverability may also depend on the applied
command state. We therefore lift the terminal set to
$\mathcal C_B^z=\{z:h_{B,j}(z)\geq0,\ j=1,\ldots,N_B\}$ and define
$h_{B,j}^{\mathrm b}(z,T):=h_{B,j}(\varphi_{\mathrm b}^{z}(T;z))$.
The future margins test sampled safety along the recovery trajectory, while
the terminal margins test whether the predicted backup state reaches
$\mathcal C_B^z$. Their gradients require the sensitivity of the backup flow
to the current augmented state; Appendix~\ref{app:backup_sensitivity} gives
the corresponding variational dynamics.

For the quadruped specialization, $\hat u_{\mathrm b}=0$, so the first-order
command dynamics drive the applied translational and yaw-rate commands toward
zero. The backup test therefore asks whether sufficient obstacle clearance
remains to dissipate the current commanded motion.

\textbf{Composite recoverability certificate.}
\label{sec:composite_lse}
The backup construction above yields the current margins $\psi_i(z)$, sampled
future margins $h_i^{\mathrm b}(z,\tau_k)$, terminal margins
$h_{B,j}^{\mathrm b}(z,T)$, and input-admissibility margins $I_d^\pm(z)$.
After fixed positive rescaling, denoted by an overbar, SAFTER combines them
using the LSE smooth minimum \citep{molnar2023composing,rabiee2025soft},
\begin{equation}
    \softmin_{\kappa}(\mathcal S)
    :=
    -\frac{1}{\kappa}\log\!\sum_{c\in\mathcal S}e^{-\kappa c},
    \qquad
    \min\mathcal S-\frac{\log|\mathcal S|}{\kappa}
    \leq\softmin_{\kappa}(\mathcal S)\leq\min\mathcal S,
    \label{eq:lse_main}
\end{equation}
where $\kappa>0$ controls the approximation sharpness. Thus,
$\softmin_{\kappa}(\mathcal S)\geq0$ implies nonnegativity of every
constituent in $\mathcal S$. We aggregate recoverability and input admissibility separately,
\begin{align*}
    h_{\mathrm{rec}}(z)
    &:=\softmin_{\kappa_B}\!\left(
        \{\bar\psi_i(z)\}_i
        \cup\{\bar h_i^{\mathrm b}(z,\tau_k)\}_{i,k}
        \cup\{\bar h_{B,j}^{\mathrm b}(z,T)\}_j
    \right),\\
    h_{\mathrm{input}}(z)
    &:=\softmin_{\kappa_U}\!\left(
        \{\bar I_d^+(z),\bar I_d^-(z)\}_{d=1}^{m}
    \right).
\end{align*}
The complete certificate is
$
    h_{\mathrm{aug}}(z)
    :=\softmin_{\kappa_O}\!\left(
        \{h_{\mathrm{rec}}(z),h_{\mathrm{input}}(z)\}
    \right).
$
Hence, $h_{\mathrm{aug}}(z)\geq0$ implies all encoded current, sampled
backup, terminal, and input-admissibility conditions. The individual margins
must still be evaluated, but their enforcement is reduced to one
differentiable scalar certificate.

\textbf{Analytical minimum-intervention projection.}
The composite certificate $h_{\mathrm{aug}}$ now plays the role of the
barrier function in the CBF condition introduced in
\eqref{eq:standard_cbf_qp}. Applied to the augmented dynamics
\eqref{eq:augmented_dynamics}, this condition becomes
\begin{equation}
    L_{\hat f}h_{\mathrm{aug}}(z)
    +L_{\hat g}h_{\mathrm{aug}}(z)\hat u
    +\alpha_{\mathrm{aug}}h_{\mathrm{aug}}(z)
    \geq0,
    \label{eq:composite_affine_condition}
\end{equation}
where $\alpha_{\mathrm{aug}}>0$ is the composite CBF gain. Defining
$\mathbf a(z):=L_{\hat g}h_{\mathrm{aug}}(z)$ and
$b(z):=L_{\hat f}h_{\mathrm{aug}}(z)
+\alpha_{\mathrm{aug}}h_{\mathrm{aug}}(z)$,
\eqref{eq:composite_affine_condition} is simply
$\mathbf a(z)\hat u+b(z)\geq0$: a single half-space in the surrogate input.
Thus, the many sampled recoverability conditions composed above no longer
appear as separate constraints in the online projection.
The actor, however, proposes the physical command $u_t^{\mathrm{nom}}$ rather
than the surrogate input $\hat u$. To recover this nominal command when the
safety constraint is inactive, we use
\begin{equation}
    \hat u_t^{\mathrm{nom}}
    =
    u_t
    +A^{-1}\dot u_t^{\mathrm{nom}}
    +A^{-1}\Sigma
    \left(u_t^{\mathrm{nom}}-u_t\right),
    \qquad \Sigma\succ0,
    \label{eq:nominal_surrogate}
\end{equation}
which makes the applied command $u$ exponentially track
$u^{\mathrm{nom}}$ under the augmented dynamics. The derivation is given in
Appendix~\ref{app:nominal_tracking}. The safety filter therefore solves
\begin{equation}
\begin{aligned}
    \hat u_t^\star
    ={}
    \arg\min_{\hat u\in\mathbb R^m}
    \frac12\|\hat u-\hat u_t^{\mathrm{nom}}\|_2^2, \quad 
    \mathrm{s.t.} \quad
    \mathbf a(z_t)\hat u+b(z_t)\geq0 .
\end{aligned}
\label{eq:closed_form_qp}
\end{equation}
Let
$\Omega_t:=\mathbf a(z_t)\hat u_t^{\mathrm{nom}}+b(z_t)$
denote the composite constraint margin of the nominal surrogate, and $\delta_t := a(z_t)a(z_t)^\top=\|a(z_t)\|_2^2$ denote the squared constraint-normal magnitude.

\begin{proposition}[Closed-form safety projection]
\label{prop:closed_form_projection}
If $\delta_t>0$, the unique solution of \eqref{eq:closed_form_qp} is
\begin{equation}
    \hat{u}_t^\star = \hat{u}_t^{\mathrm{nom}} + \Delta\hat{u}_t
    =
    \hat{u}_t^{\mathrm{nom}}
    +
    \mathbf a(z_t)^\top
    \frac{\max\{0,-\Omega_t\}}{\delta_t}
    .
    \label{eq:closed_form_solution}
\end{equation}
\end{proposition}
If $\Omega_t\geq0$, the nominal surrogate already satisfies the composite
condition and is unchanged; otherwise,
\eqref{eq:closed_form_solution} applies the minimum-Euclidean correction along the certificate normal. The projection is exact for the resulting
half-space, and approximation enters only through finite backup sampling and the
LSE construction preceding it. The detailed proof of Proposition \ref{prop:closed_form_projection} is included in Appendix~\ref{app:kkt_derivation}.

\begin{theorem}[Continuous-time safety and recoverability]
\label{thm:continuous_composite_safety}
Under Assumption~\ref{ass:continuous_safety_conditions}, suppose
$z(0)\in\mathcal C_{\mathrm{aug}}$ and the composite CBF condition
\eqref{eq:composite_affine_condition} is enforced continuously. Then
$\mathcal C_{\mathrm{aug}}=\{z:h_{\mathrm{aug}}(z)\geq0\}$ is forward
invariant for the continuous augmented model. Consequently, the safety-model
state remains in the original safe set $\mathcal C_S$, the applied command
remains within its prescribed bounds, and the backup trajectory satisfies
the encoded safety conditions at the sampled future times and reaches the
terminal recovery set at the end of the horizon.
\end{theorem}

The result applies to the continuous safety model and to backup safety at the
selected horizon samples. The finite-rate implementation, inter-sample
backup behavior, sensing updates, and reduced-model assumptions are discussed
in Appendices~\ref{app:finite_rate_exp} and~\ref{app:guarantee_scope}.

\begin{wrapfigure}{r}{0.5\columnwidth}
    \centering
    \includegraphics[width=1\linewidth]{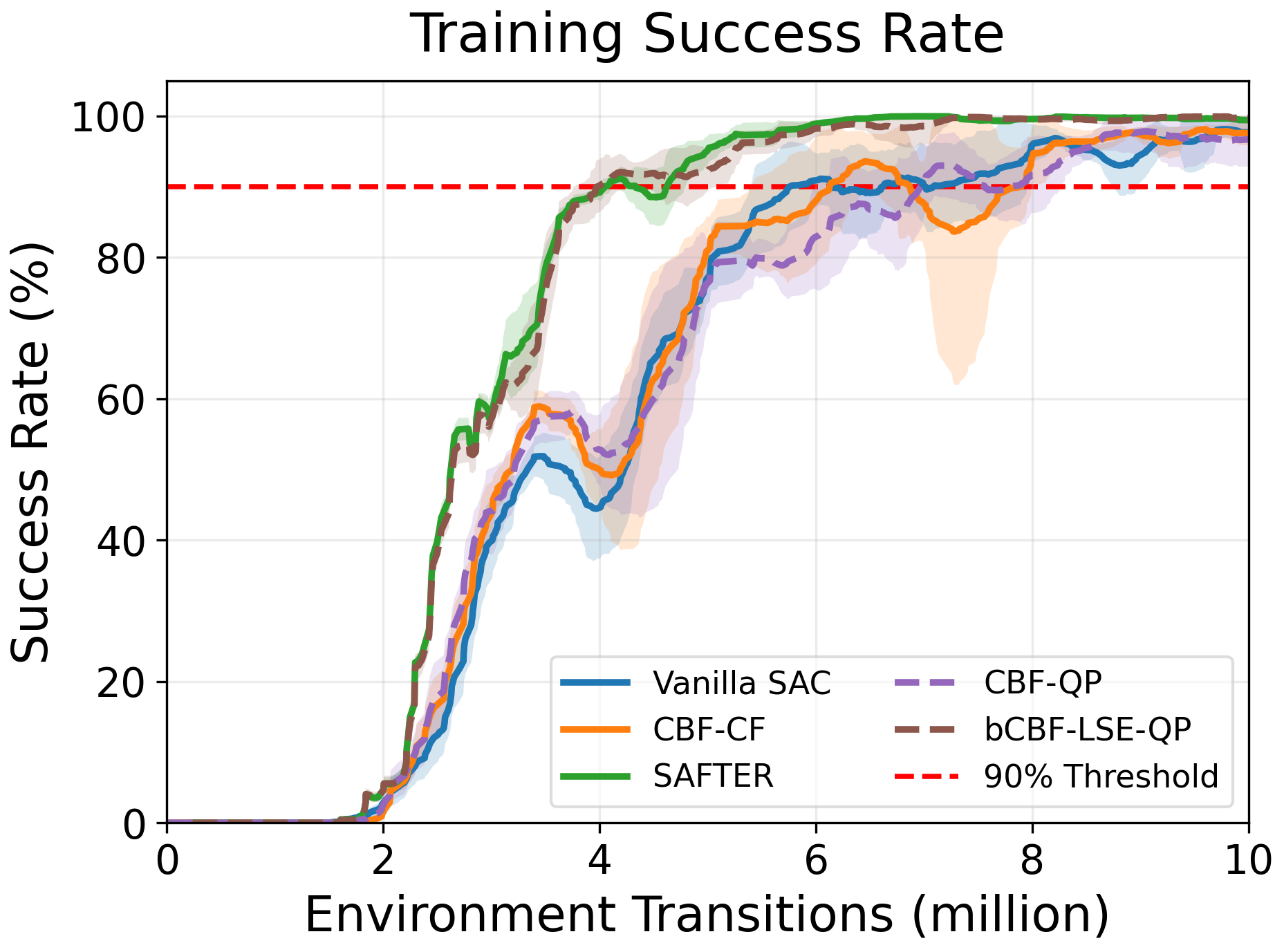}
    \caption{Training success rates of Vanilla SAC, CBF--QP/CBF--CF, bCBF--LSE--QP, and SAFTER under the same training environments.}
    \label{fig:training_SR}
\end{wrapfigure}
\textbf{SAFTER--SAC integration.}
\label{sec:sac_integration}
As illustrated in Figure~\ref{fig:architecture}, SAFTER places the
recoverability-aware safety map $\mathcal F_\eta$ between the SAC actor and
the environment. At each step, the actor proposes $u_t^{\mathrm{nom}}$,
the safety map produces the executed command $u_t^{\mathrm{safe}}$, and the
environment transition is generated by this filtered action. Accordingly,
$u_t^{\mathrm{safe}}$ is stored in replay and evaluated by the critics, so
value learning is based on the action that actually generated the transition.
The same filtering is applied to next actions when constructing SAC targets.

During the actor update, the critic evaluates $Q(o_t,\mathcal F_\eta(\xi_t,u_t^{\mathrm{nom}}))$, and gradients propagate through the same finite-rate map used during execution: nominal-surrogate construction, analytical half-space projection, command update, and component-wise clipping.
This map is differentiable away from projection and clipping switching boundaries, allowing actor updates to account for how safety interventions modify nominal commands.
Target filtering uses the same forward map without gradient propagation, while the SAC entropy term remains defined on the nominal policy.
The full SAC implementation details are reported in Appendix~\ref{app:sac_details}.


\section{Experiments and Analysis}
\begin{figure*}[t]
\centering
\begin{subfigure}[c]{0.19\textwidth}
    \centering
    \includegraphics[width=\linewidth]{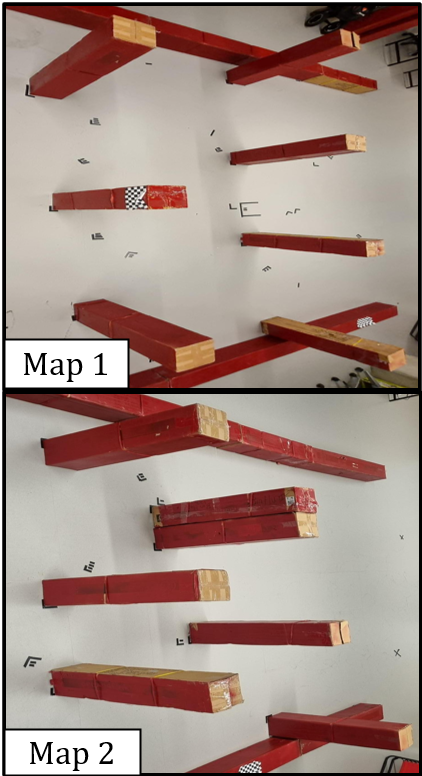}
    \caption{Real-world evaluation maps.}
    \label{fig:real_maps}
\end{subfigure}
\hfill
\begin{subfigure}[c]{0.79\textwidth}
    \centering
    \includegraphics[width=\linewidth]{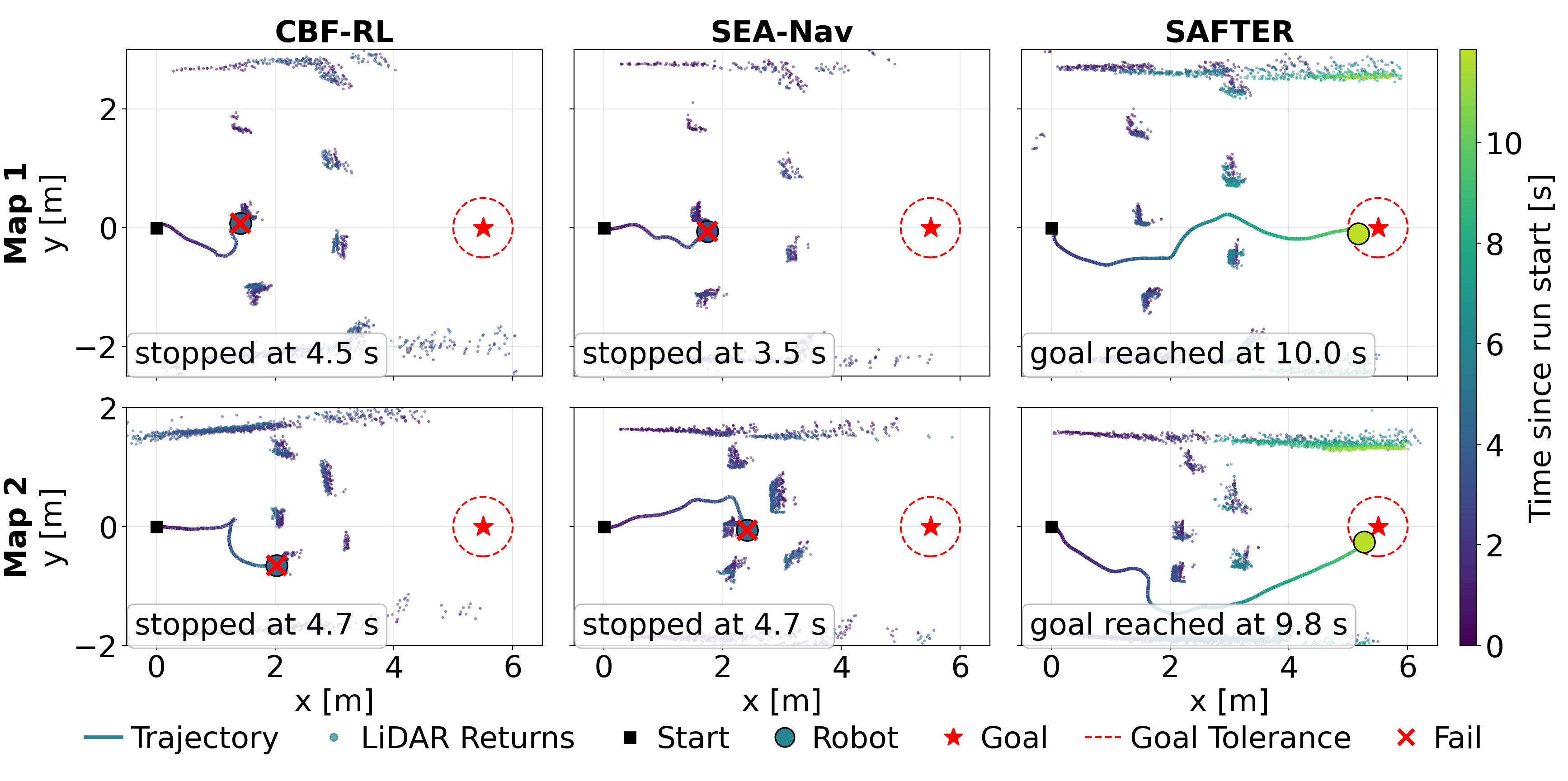}
    \caption{Trajectory sample and LiDAR readings of the real robot experiments.}
    \label{fig:real_trajectories}
\end{subfigure}
\caption{Real-world map with two different obstacle configurations, and trajectory sample and LiDAR readings taken from the Unitree Go2 on-board state estimator and LiDAR.}
\end{figure*}

\label{sec:experiments}
The experiments are organized to test the three mechanisms isolated in the problem formulation.
\textbf{RQ1} asks whether finite-horizon backup recoverability, rather than current-state filtering alone, improves transition-based learning efficiency, and whether its analytical realization reduces safety-filter latency relative to QP-based synthesis.
\textbf{RQ2} asks whether the resulting policy--filter system remains effective as obstacle geometry shifts away from the training maps.
\textbf{RQ3} asks whether learning through the exact analytical projection Jacobian matters when the forward safety-filter mapping is unchanged.

The baselines and SAFTER are evaluated in MuJoCo~3.3.6 using a Unitree Go2 with a frozen $50\,\mathrm{Hz}$ low-level locomotion policy.
The high-level policy operates at $10\,\mathrm{Hz}$ using local LiDAR, goal, measured-velocity, and command-history observations.
SAC-based methods use 256 parallel environments, five independent training seeds, and a maximum budget of $20.48$M transitions per seed.
In both training and evaluation, failure denotes episode termination caused by obstacle collision, falling, or contact with the environment boundary wall.
The bCBF--LSE--QP baseline uses the same composite half-space, backup model, terminal set, command dynamics, and finite-rate realization as SAFTER.

\begin{table}[t]
\centering
\caption{Performance on the training-map distribution with each method's runtime safety filter disabled or enabled. SR, TR, and FR denote success, timeout, and failure rates, respectively. Values are mean $\pm$ standard deviation over evaluation seeds and maps.}
\label{tab:training_maps_filter_ablation}
\small
\setlength{\tabcolsep}{4.0pt}
\renewcommand{\arraystretch}{1.08}
\begin{tabular}{lcccccc}
\toprule
& \multicolumn{3}{c}{\textbf{Filter disabled}} & \multicolumn{3}{c}{\textbf{Filter enabled}} \\
\cmidrule(lr){2-4} \cmidrule(lr){5-7}
\textbf{Method}
& \textbf{SR $\uparrow$} & \textbf{TR $\downarrow$} & \textbf{FR $\downarrow$}
& \textbf{SR $\uparrow$} & \textbf{TR $\downarrow$} & \textbf{FR $\downarrow$} \\
\midrule

SEA--Nav\textsuperscript{1}
& $89.7{\pm}2.4$ & $0.0{\pm}0.0$ & $10.3{\pm}2.4$
& $91.2{\pm}2.9$ & $0.0{\pm}0.0$ & $8.8{\pm}2.9$ \\

CBF--RL\textsuperscript{2}
& $84.3{\pm}8.2$ & $5.9{\pm}5.2$ & $9.8{\pm}7.1$
& $84.8{\pm}7.4$ & $5.9{\pm}5.2$ & $9.3{\pm}6.6$ \\

CBF--CF\textsuperscript{3}
& $76.6{\pm}13.8$ & $2.7{\pm}6.3$ & $20.7{\pm}15.8$
& $90.1{\pm}8.2$ & $1.9{\pm}4.5$ & $8.0{\pm}8.5$ \\

SAFTER\textsuperscript{4}
& $\mathbf{95.1{\pm}7.1}$ & $0.0{\pm}0.0$ & $\mathbf{4.9{\pm}7.1}$
& $\mathbf{99.7{\pm}1.2}$ & $0.0{\pm}0.0$ & $\mathbf{0.3{\pm}1.2}$ \\

\bottomrule
\end{tabular}
\par
{\scriptsize
\textit{Refs.:}
\textsuperscript{1}~\citep{chen2026sea};
\textsuperscript{2}~\citep{yang2025cbf};
\textsuperscript{3}~\citep{rabiee2024closed};
\textsuperscript{4}~\textbf{Ours}.
}

\vspace{-3mm}
\end{table}

\textbf{RQ1: Learning Efficiency and Policy Internalization.}
For the controlled convergence study, sample efficiency is measured by the first attainment of $90\%$ training success for each seed, using a 20-sample trailing rolling mean; the threshold is counted at its first crossing and is not required to persist.
As shown in Figure~\ref{fig:training_SR}, the two backup-recoverability methods converge earliest: SAFTER reaches the threshold after $4.06{\pm}0.08$M transitions and bCBF--LSE--QP after $4.13{\pm}0.43$M.
SAFTER therefore requires $30.6\%$ fewer transitions than Vanilla SAC and $34.1\%$--$35.1\%$ fewer than the current-state CBF baselines.
The similar convergence of SAFTER and bCBF--LSE--QP indicates that the sample-efficiency gain is primarily associated with predictive backup recoverability rather than the analytical realization alone.

Recoverability-aware training also appears to shape the learned nominal policy rather than relying solely on runtime filtering.
To examine this effect, we evaluate the trained policies on the training-map distribution with their runtime safety filters enabled and disabled.
Among the methods reported in Table~\ref{tab:training_maps_filter_ablation}, SAFTER achieves $95.1\%$ success without the runtime filter and $99.7\%$ with it.
The strong filter-disabled performance suggests that recoverability-aware training promotes nominal behavior compatible with the prescribed backup strategy.

\begin{table}[t]
\centering
\caption{
Safety-filter runtime comparison.
Total latency includes certificate construction and synthesis, and speedup is
measured relative to the direct bCBF--QP.
Values are mean $\pm$ standard deviation over 300 timed repetitions.
}
\label{tab:runtime_benchmark}

\small
\setlength{\tabcolsep}{3.0pt}
\renewcommand{\arraystretch}{1.08}

\begin{tabular}{lcccccc}
\toprule
& \multicolumn{3}{c}{\textbf{Batch = 1}}
& \multicolumn{3}{c}{\textbf{Batch = 100}} \\
\cmidrule(lr){2-4}
\cmidrule(lr){5-7}

\textbf{Method}
& \shortstack{\textbf{Total}\\\textbf{(ms)}}
& \shortstack{\textbf{Synthesis}\\\textbf{(ms)}}
& \shortstack{\textbf{Speedup}\\\textbf{($\times$)}}
& \shortstack{\textbf{Total}\\\textbf{(ms)}}
& \shortstack{\textbf{Synthesis}\\\textbf{(ms)}}
& \shortstack{\textbf{Speedup}\\\textbf{($\times$)}} \\
\midrule

bCBF--QP
& $30.546{\pm}0.324$
& $28.369{\pm}0.194$
& $1.00$
& $2504.789{\pm}22.463$
& $2487.383{\pm}16.949$
& $1.00$ \\

bCBF--LSE--QP
& $5.732{\pm}0.456$
& $3.342{\pm}0.405$
& $5.33$
& $7.084{\pm}0.569$
& $4.588{\pm}0.481$
& $353.58$ \\

SAFTER
& $\mathbf{2.495{\pm}0.312}$
& $\mathbf{0.102{\pm}0.011}$
& $\mathbf{12.24}$
& $\mathbf{2.529{\pm}0.283}$
& $\mathbf{0.106{\pm}0.014}$
& $\mathbf{990.43}$ \\

\bottomrule
\end{tabular}
\vspace{-3mm}
\end{table}

To evaluate computational efficiency, Table~\ref{tab:runtime_benchmark} compares SAFTER with the matched bCBF--LSE--QP and the direct unaggregated
bCBF--QP.
All methods are evaluated on identical fixed filter inputs using 30 warm-up calls and 300 timed repetitions on an NVIDIA RTX 4090 GPU with CUDA synchronization around each timed region.
SAFTER achieves the lowest total latency at both batch sizes, providing $12.24\times$ and $990.43\times$ lower latency than the direct bCBF--QP at batch sizes $1$ and $100$, respectively.
Relative to the matched bCBF--LSE--QP, SAFTER is $2.30\times$ and $2.80\times$ faster.

\begin{figure}[t]
    \centering
    \includegraphics[width=0.75\linewidth]{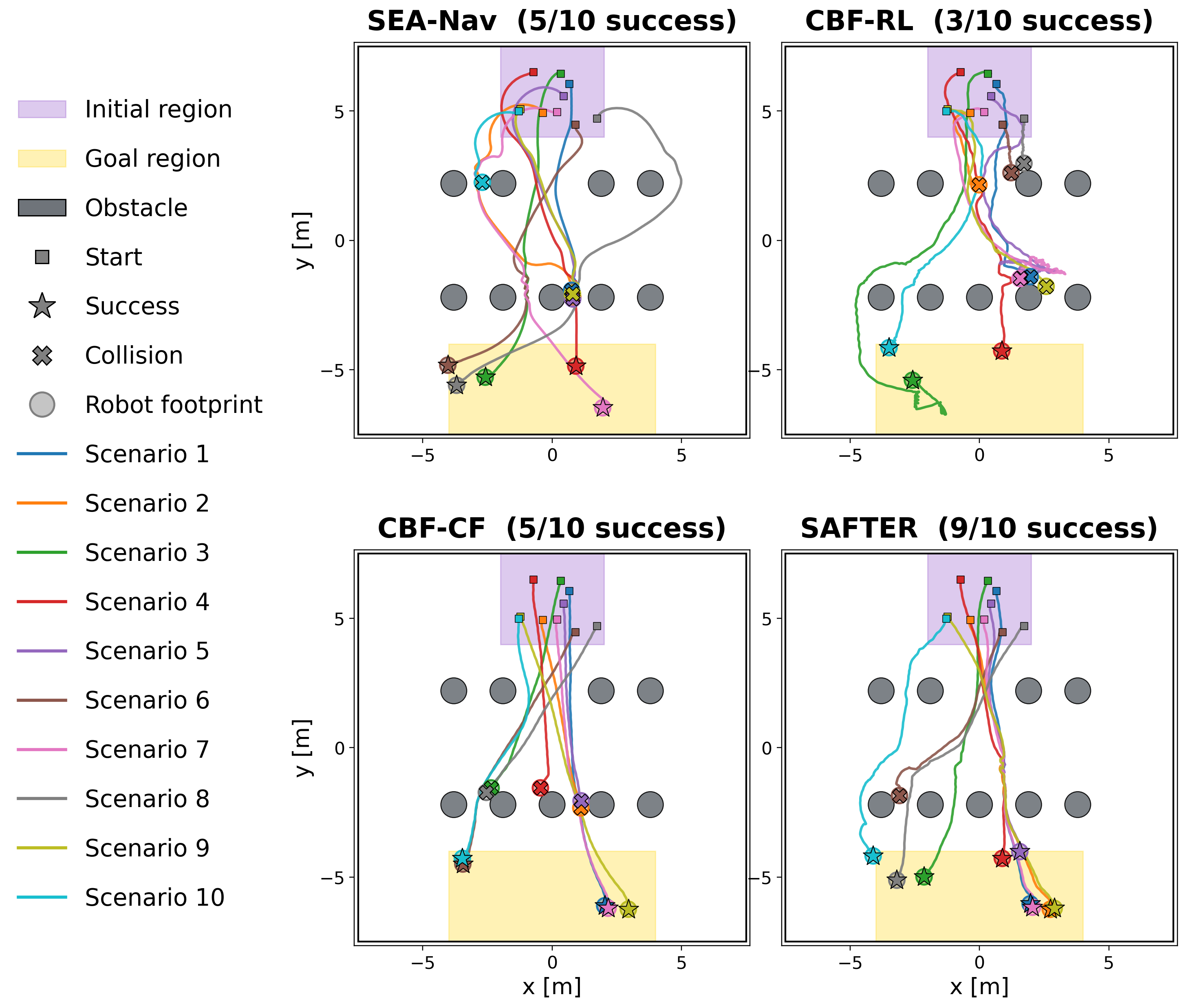}
    \caption{Representative trajectories in the simulated Medium-Shift environment.
    Purple and yellow regions denote the initial and goal regions, squares mark initial positions, stars successful terminations, crosses collision terminations, and faint circles indicate the robot footprint.
    }
    \label{fig:sample_trajectory}
    \vspace{-3mm}
\end{figure}

\textbf{RQ2: Robustness under Geometric Distribution Shift.}
Predictive backup recoverability supports robust navigation under geometric shift in simulation and real-world deployment.
Each fixed checkpoint is evaluated over 500 episodes per unseen simulation environment and 30 trials per real-world map.
Trajectories in Figures~\ref{fig:sample_trajectory} and~\ref{fig:real_trajectories} illustrate navigation behavior in simulation and real-world experiments.
The real-world and simulated examples show SAFTER taking wider detours with greater obstacle clearance, consistent with preserving room for backup recovery.
The control actions generated by SAFTER produce a safer navigation in an environment under distribution shift, as shown by the overall results reported in Table~\ref{tab:eval_filter_active}, where SAFTER achieves $98.8\%$, $88.6\%$, and $61.4\%$ success on Low-, Medium-, and High-Shift, outperforming the baseline methods.
On the two real-world layouts in Figure~\ref{fig:real_maps}, SAFTER achieves $73.3\%$ and $76.7\%$ success, exceeding the strongest evaluated baselines by $10.0$ and $60.0$ percentage points, respectively.
Moreover, additional experimental details and results are reported in Appendix~\ref{app:experiments}.

\begin{table}[t]
    \centering
    \caption{Navigation performance in simulation and real-world experiments.
    SR, TR, and FR denote success, timeout, and failure rates (\%), respectively.
    Simulation values are mean $\pm$ standard deviation across 25 evaluation-seed rates.
    Real-world success rates are measured over 30 trials per map.
    Dashes indicate methods not evaluated on hardware.}
    \label{tab:eval_filter_active}

    \scriptsize
    \setlength{\tabcolsep}{2pt}
    \renewcommand{\arraystretch}{1.10}

    \begin{adjustbox}{max width=\linewidth}
    \begin{tabular}{
        @{}l
        ccc
        @{\hspace{5pt}}ccc
        @{\hspace{5pt}}ccc
        @{\hspace{7pt}}cc@{}
    }
        \toprule
        & \multicolumn{9}{c}{\textbf{Simulation}}
        & \multicolumn{2}{c}{\textbf{Real world}} \\
        \cmidrule(lr){2-10}
        \cmidrule(lr){11-12}

        & \multicolumn{3}{c}{\textbf{Low-Shift}}
        & \multicolumn{3}{c}{\textbf{Medium-Shift}}
        & \multicolumn{3}{c}{\textbf{High-Shift}}
        & \textbf{Map 1}
        & \textbf{Map 2} \\
        \cmidrule(lr){2-4}
        \cmidrule(lr){5-7}
        \cmidrule(lr){8-10}

        \textbf{Method}
        & \textbf{SR $\uparrow$}
        & \textbf{TR $\downarrow$}
        & \textbf{FR $\downarrow$}
        & \textbf{SR $\uparrow$}
        & \textbf{TR $\downarrow$}
        & \textbf{FR $\downarrow$}
        & \textbf{SR $\uparrow$}
        & \textbf{TR $\downarrow$}
        & \textbf{FR $\downarrow$}
        & \textbf{SR $\uparrow$}
        & \textbf{SR $\uparrow$} \\
        \midrule

        SEA--Nav\textsuperscript{1}
        & $93.8{\pm}6.7$ & $0.0{\pm}0.0$ & $6.2{\pm}6.7$
        & $50.6{\pm}10.9$ & $0.0{\pm}0.0$ & $49.4{\pm}10.9$
        & $0.0{\pm}0.0$ & $0.0{\pm}0.0$ & $100.0{\pm}0.0$
        & $56.7$ & $13.3$ \\

        CBF--RL\textsuperscript{2}
        & $61.2{\pm}10.7$ & $13.6{\pm}6.7$ & $25.2{\pm}9.8$
        & $46.2{\pm}11.6$ & $17.2{\pm}7.6$ & $36.6{\pm}9.2$
        & $31.0{\pm}10.0$ & $15.4{\pm}8.9$ & $53.6{\pm}12.2$
        & $63.3$ & $0.0$ \\

        \midrule

        Vanilla SAC\textsuperscript{3}
        & $34.2{\pm}10.1$ & $0.0{\pm}0.0$ & $65.8{\pm}10.1$
        & $20.8{\pm}6.9$ & $0.0{\pm}0.0$ & $79.2{\pm}6.9$
        & $17.4{\pm}8.3$ & $0.0{\pm}0.0$ & $82.6{\pm}8.3$
        & $10.0$ & $0.0$ \\

        CBF--QP\textsuperscript{4}
        & $58.8{\pm}9.4$ & $0.2{\pm}1.0$ & $41.0{\pm}9.2$
        & $43.2{\pm}11.0$ & $0.0{\pm}0.0$ & $56.8{\pm}11.0$
        & $18.0{\pm}6.9$ & $0.0{\pm}0.0$ & $82.0{\pm}6.9$
        & -- & -- \\

        CBF--CF\textsuperscript{4}
        & $69.0{\pm}11.5$ & $0.0{\pm}0.0$ & $31.0{\pm}11.5$
        & $48.4{\pm}8.4$ & $0.0{\pm}0.0$ & $51.6{\pm}8.4$
        & $33.8{\pm}12.2$ & $0.0{\pm}0.0$ & $66.2{\pm}12.2$
        & $56.7$ & $16.7$ \\

        bCBF--LSE--QP\textsuperscript{5}
        & $\mathbf{99.4{\pm}1.7}$ & $0.0{\pm}0.0$
        & $\mathbf{0.6{\pm}1.7}$
        & $75.0{\pm}9.7$ & $0.0{\pm}0.0$ & $25.0{\pm}9.7$
        & $\mathbf{67.4{\pm}11.7}$ & $0.0{\pm}0.0$
        & $\mathbf{32.6{\pm}11.7}$
        & -- & -- \\

        SAFTER\textsuperscript{5}
        & $98.8{\pm}2.2$ & $0.0{\pm}0.0$ & $1.2{\pm}2.2$
        & $\mathbf{88.6{\pm}7.3}$ & $0.0{\pm}0.0$
        & $\mathbf{11.4{\pm}7.3}$
        & $61.4{\pm}10.2$ & $0.0{\pm}0.0$ & $38.6{\pm}10.2$
        & $\mathbf{73.3}$ & $\mathbf{76.7}$ \\

        \bottomrule
    \end{tabular}
    \end{adjustbox}

    \par\smallskip
    {\scriptsize
    \textit{Refs.:}
    \textsuperscript{1}~\citep{chen2026sea};
    \textsuperscript{2}~\citep{yang2025cbf};
    \textsuperscript{3}~\citep{haarnoja2018soft};
    \textsuperscript{4}~\citep{rabiee2024closed};
    \textsuperscript{5}~\textbf{Ours}.
    }
\end{table}

\textbf{RQ3: Why Does SAFTER Learn Differently?}
Accurate differentiation through the analytical projection has a substantially larger and more consistent effect on robustness than discouraging filter intervention.
As shown in Table~\ref{tab:bcbf_gradient_reward_ablation_full}, replacing the Exact projection Jacobian with the Identity approximation reduces success by up to $85.4$ percentage points across the evaluated distribution shifts.
In contrast, removing the correction reward has a mixed effect: the Exact-Jacobian policy improves from $88.6\%$ to $94.4\%$ on Medium-Shift and from $61.4\%$ to $76.2\%$ on High-Shift, while decreasing from $98.8\%$ to $95.6\%$ on Low-Shift.
Because the forward safety-filter map is unchanged within each Exact/Identity pair, these results provide evidence that propagating the analytical projection Jacobian materially improves robustness, whereas minimizing filter intervention is not uniformly beneficial under larger geometric shifts.

\begin{table}[h]
\centering
\caption{
Ablation of the projection Jacobian and correction reward on unseen environments.
Exact and Identity use the same forward safety-filter map and differ only in the backward differentiation through the analytical projection.
Values are mean $\pm$ standard deviation across 25 evaluation-seed rates.
}
\label{tab:bcbf_gradient_reward_ablation_full}

\setlength{\tabcolsep}{2.4pt}
\renewcommand{\arraystretch}{1.08}

\resizebox{\linewidth}{!}{
\begin{tabular}{cc ccc ccc ccc}
\toprule

\multicolumn{2}{c}{\textbf{Ablation Setting}}
& \multicolumn{3}{c}{\textbf{Low-Shift}}
& \multicolumn{3}{c}{\textbf{Medium-Shift}}
& \multicolumn{3}{c}{\textbf{High-Shift}}
\\

\cmidrule(lr){1-2}
\cmidrule(lr){3-5}
\cmidrule(lr){6-8}
\cmidrule(lr){9-11}

\shortstack{\textbf{Projection}\\\textbf{Jacobian}}
& \shortstack{\textbf{Correction}\\\textbf{Reward}}
& \textbf{SR}$\uparrow$ & \textbf{TR}$\downarrow$ & \textbf{FR}$\downarrow$
& \textbf{SR}$\uparrow$ & \textbf{TR}$\downarrow$ & \textbf{FR}$\downarrow$
& \textbf{SR}$\uparrow$ & \textbf{TR}$\downarrow$ & \textbf{FR}$\downarrow$
\\

\midrule

Exact
& On
& $\mathbf{98.8{\pm}2.2}$ & $0.0{\pm}0.0$ & $1.2{\pm}2.2$
& $88.6{\pm}7.3$ & $0.0{\pm}0.0$ & $11.4{\pm}7.3$
& $61.4{\pm}10.2$ & $0.0{\pm}0.0$ & $38.6{\pm}10.2$
\\

Identity
& On
& $89.8{\pm}6.0$ & $0.0{\pm}0.0$ & $10.2{\pm}6.0$
& $41.6{\pm}9.3$ & $0.0{\pm}0.0$ & $58.4{\pm}9.3$
& $6.2{\pm}4.4$ & $0.0{\pm}0.0$ & $93.8{\pm}4.4$
\\

Exact
& Off
& $95.6{\pm}4.9$ & $4.2{\pm}4.9$ & $\mathbf{0.2{\pm}1.0}$
& $\mathbf{94.4{\pm}4.4}$ & $3.8{\pm}3.9$ & $\mathbf{1.8{\pm}3.5}$
& $\mathbf{76.2{\pm}9.2}$ & $3.2{\pm}3.8$ & $\mathbf{20.6{\pm}8.3}$
\\

Identity
& Off
& $76.8{\pm}7.1$ & $0.0{\pm}0.0$ & $23.2{\pm}7.1$
& $9.0{\pm}6.1$ & $0.0{\pm}0.0$ & $91.0{\pm}6.1$
& $2.0{\pm}3.5$ & $0.0{\pm}0.0$ & $98.0{\pm}3.5$
\\

\bottomrule
\end{tabular}
}
\vspace{-3mm}
\end{table}

\section{Conclusion}
We introduced SAFTER, a recoverability-aware safe-RL framework that combines command-state augmentation, hierarchical LSE composition, and an exact minimum-Euclidean projection into a differentiable training map.
Experiments show improved transition-based sample efficiency and strong performance under geometric distribution shift in both simulated and real-world environments, while the Jacobian ablation highlights the importance of differentiating through the analytical safety-filter map.
Future work will consider uneven terrain and dynamic obstacles.

\subsection*{Acknowledgements}
This work was supported by the NVIDIA Academic Grant Program.



\bibliography{iclr2027_conference}
\bibliographystyle{iclr2027_conference}

\clearpage
\appendix

\section{Experimental Protocols and Additional Results}
\label{app:experiments}

\subsection{Training and Evaluation Maps}
\label{app:train_eval_maps}
The simulated quadruped experiments use a fixed $15~\si{\meter}\times15~\si{\meter}$ workspace.
The training bank comprises five hand-crafted maps with cylindrical obstacles of radius $0.6~\si{\meter}$ and 200 predefined initial--goal scenarios per map (Figure~\ref{fig:training_maps}), varying obstacle arrangements and start--goal orientations.
The policy observes local planar LiDAR with a maximum range of $5~\si{\meter}$, requiring navigation decisions as previously unseen geometry enters the sensing region.
The three unseen environments in Figure~\ref{fig:unseen_maps} introduce progressively larger geometric shifts:
\textit{Low-Shift} preserves the training maps' row-based structure while increasing obstacle count and reversing the initial and goal regions;
\textit{Medium-Shift} retains a similar obstacle count but narrows the passages between rows;
and \textit{High-Shift} introduces different obstacle radii, a square obstacle, and tighter inter-obstacle clearances than Medium-Shift.
All evaluation maps retain the protocol's minimum wall clearance for admissible initial and goal regions.

\begin{figure*}[h]
\centering

\begin{subfigure}[t]{0.31\textwidth}
    \centering
    \includegraphics[width=\linewidth]{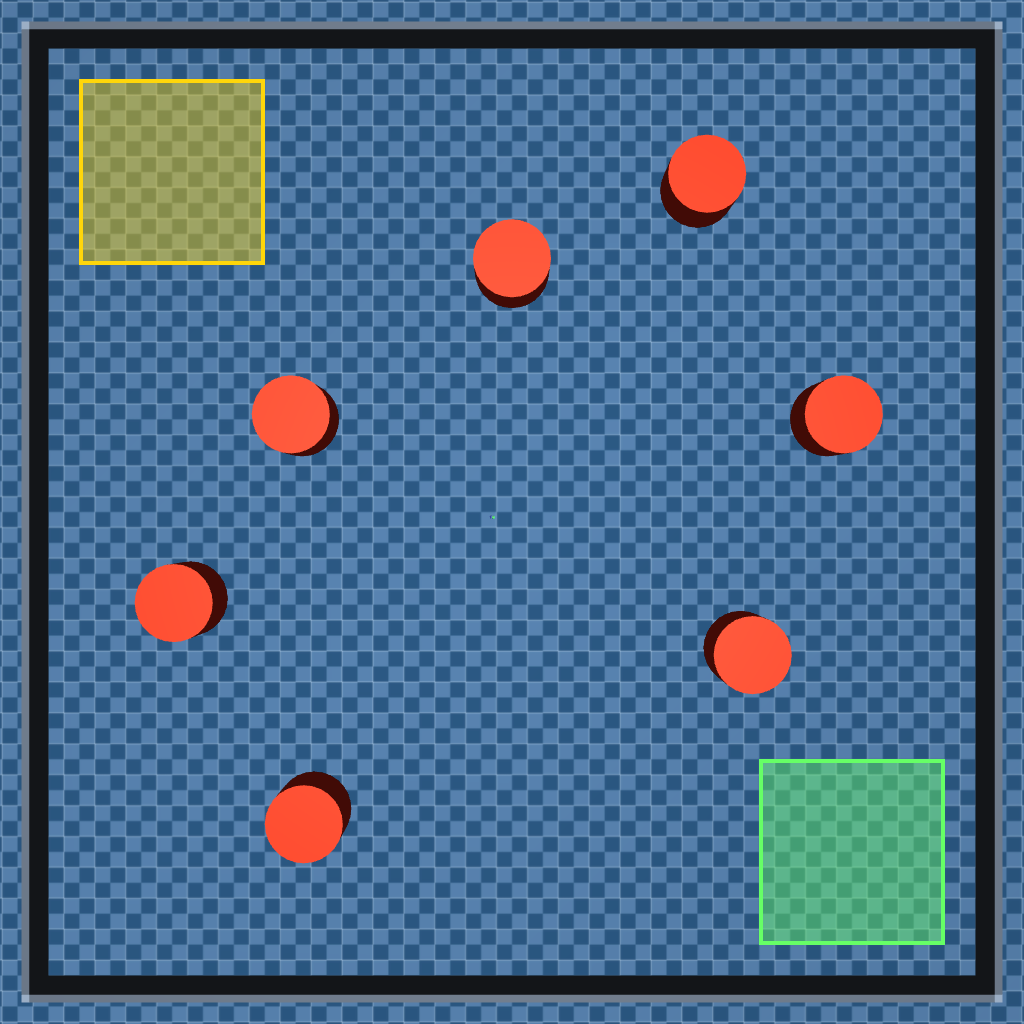}
    \caption{Training Map 1}
    \label{fig:train_map_1}
\end{subfigure}
\hfill
\begin{subfigure}[t]{0.31\textwidth}
    \centering
    \includegraphics[width=\linewidth]{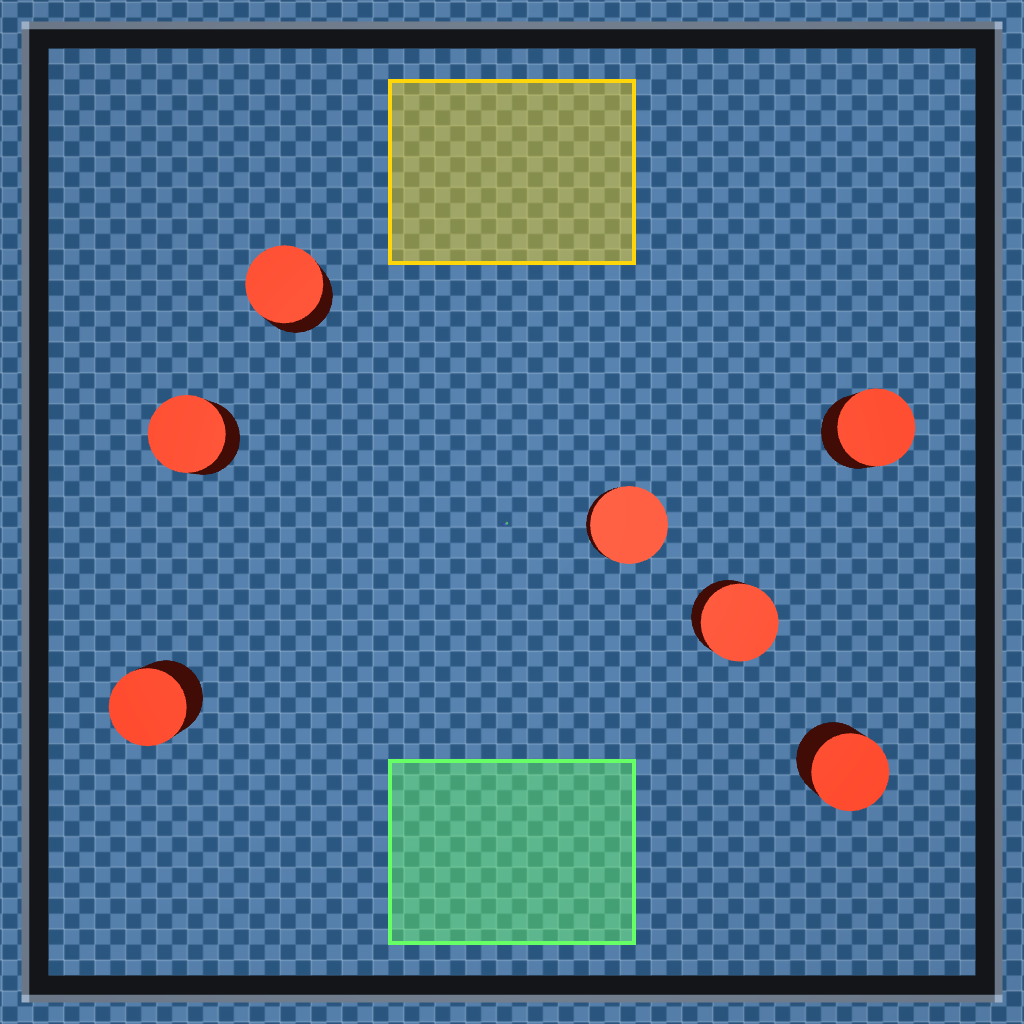}
    \caption{Training Map 2}
    \label{fig:train_map_2}
\end{subfigure}
\hfill
\begin{subfigure}[t]{0.31\textwidth}
    \centering
    \includegraphics[width=\linewidth]{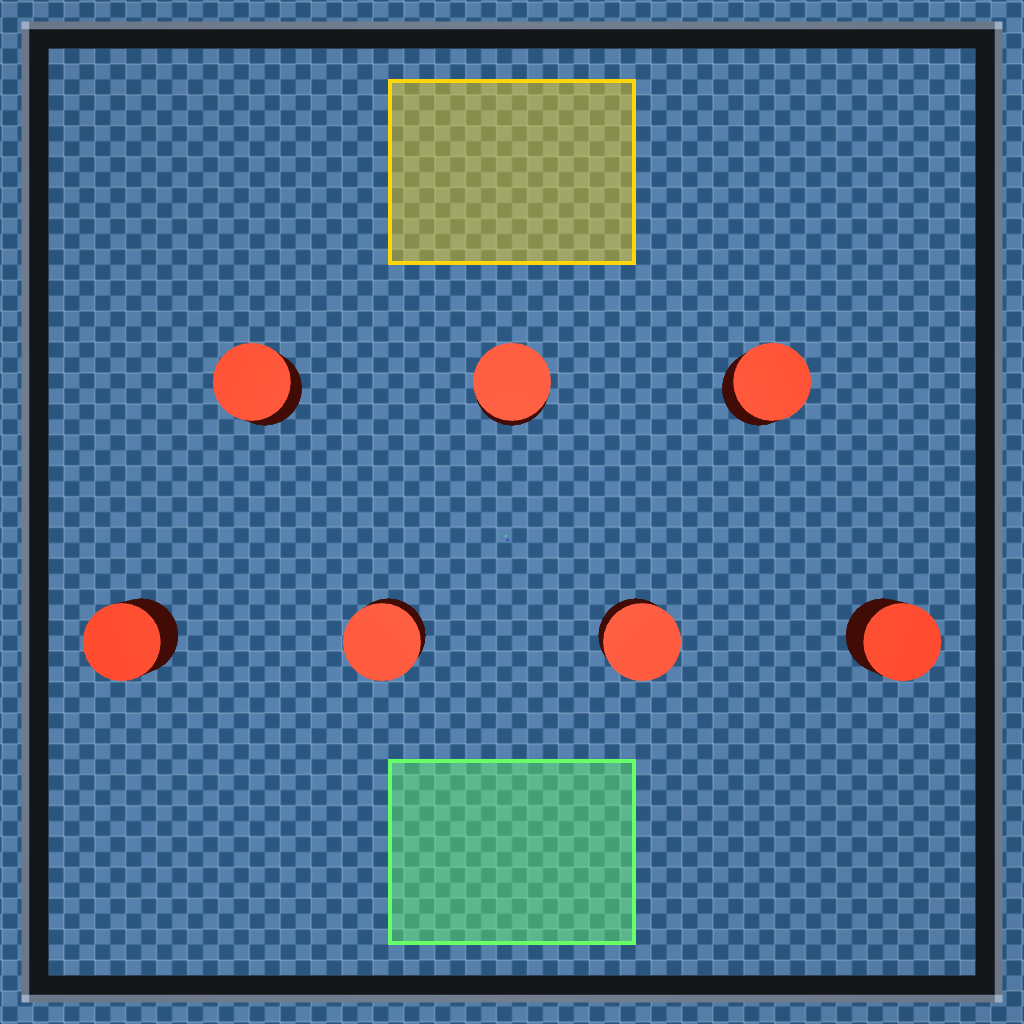}
    \caption{Training Map 3}
    \label{fig:train_map_3}
\end{subfigure}

\vspace{0.5em}

\begin{subfigure}[t]{0.31\textwidth}
    \centering
    \includegraphics[width=\linewidth]{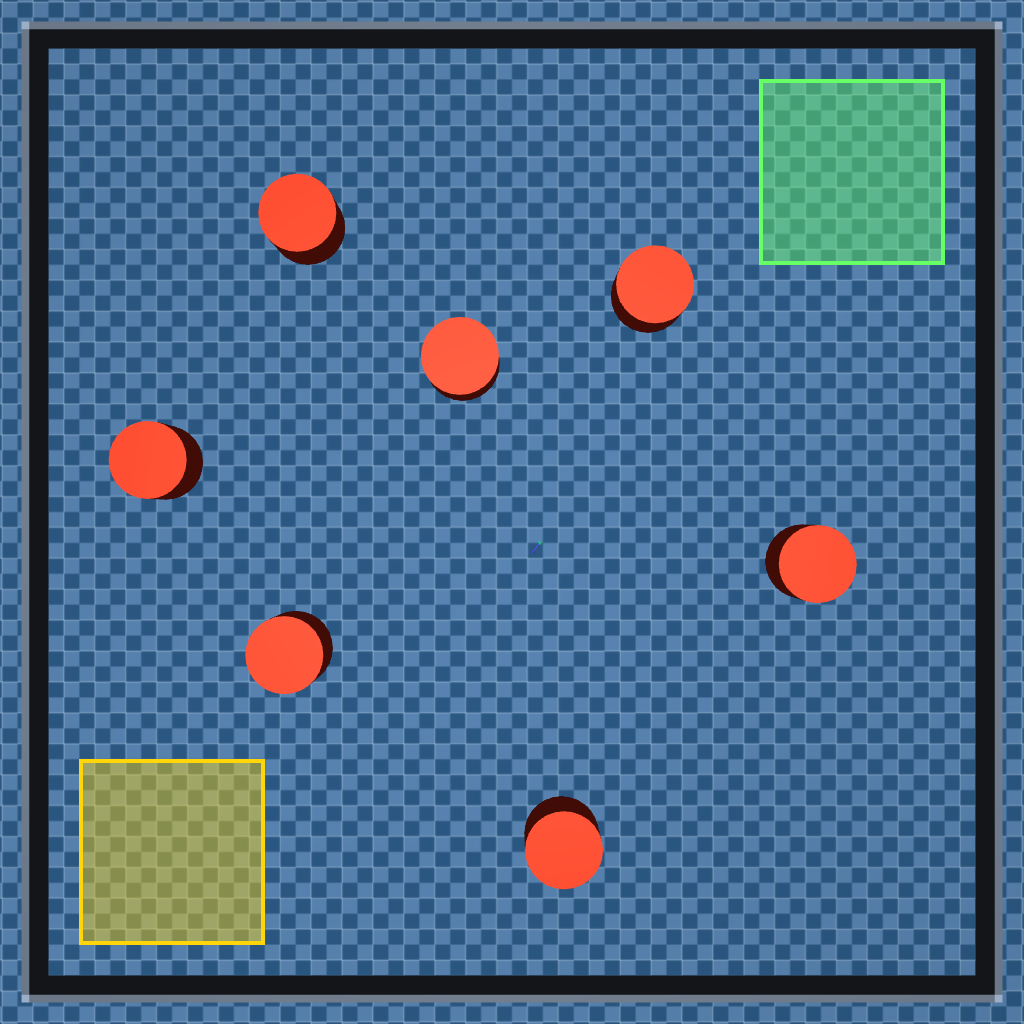}
    \caption{Training Map 4}
    \label{Figures/Experiments/Maps/training_4.png}
\end{subfigure}
\hspace{0.04\textwidth}
\begin{subfigure}[t]{0.31\textwidth}
    \centering
    \includegraphics[width=\linewidth]{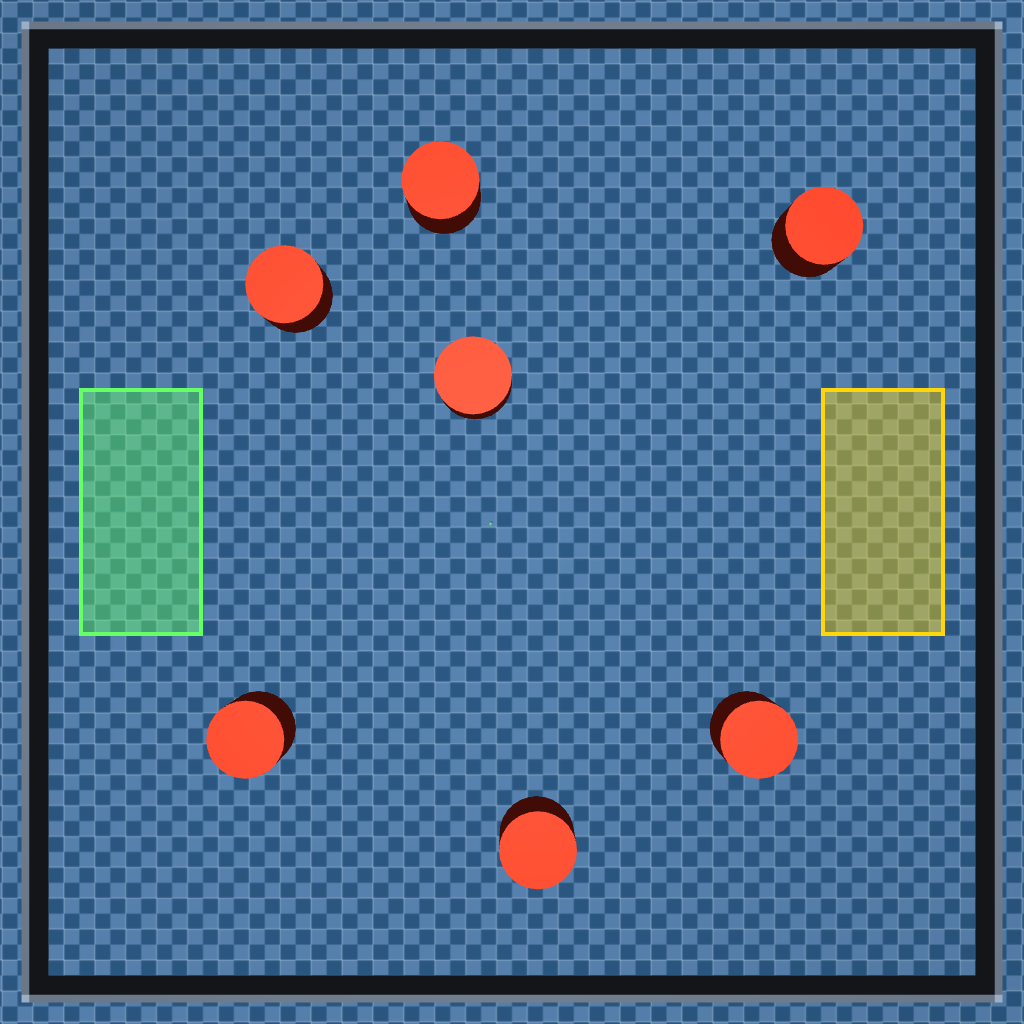}
    \caption{Training Map 5}
    \label{fig:train_map_5}
\end{subfigure}

\caption{Five hand-crafted training maps with fixed obstacle configurations and predefined initial and goal regions. Green and yellow regions indicate the initial and goal regions, respectively, while obstacles are shown in red.}
\label{fig:training_maps}
\end{figure*}

\begin{figure*}[h]
\centering
\begin{subfigure}[t]{0.32\textwidth}
    \centering
    \includegraphics[width=\linewidth]{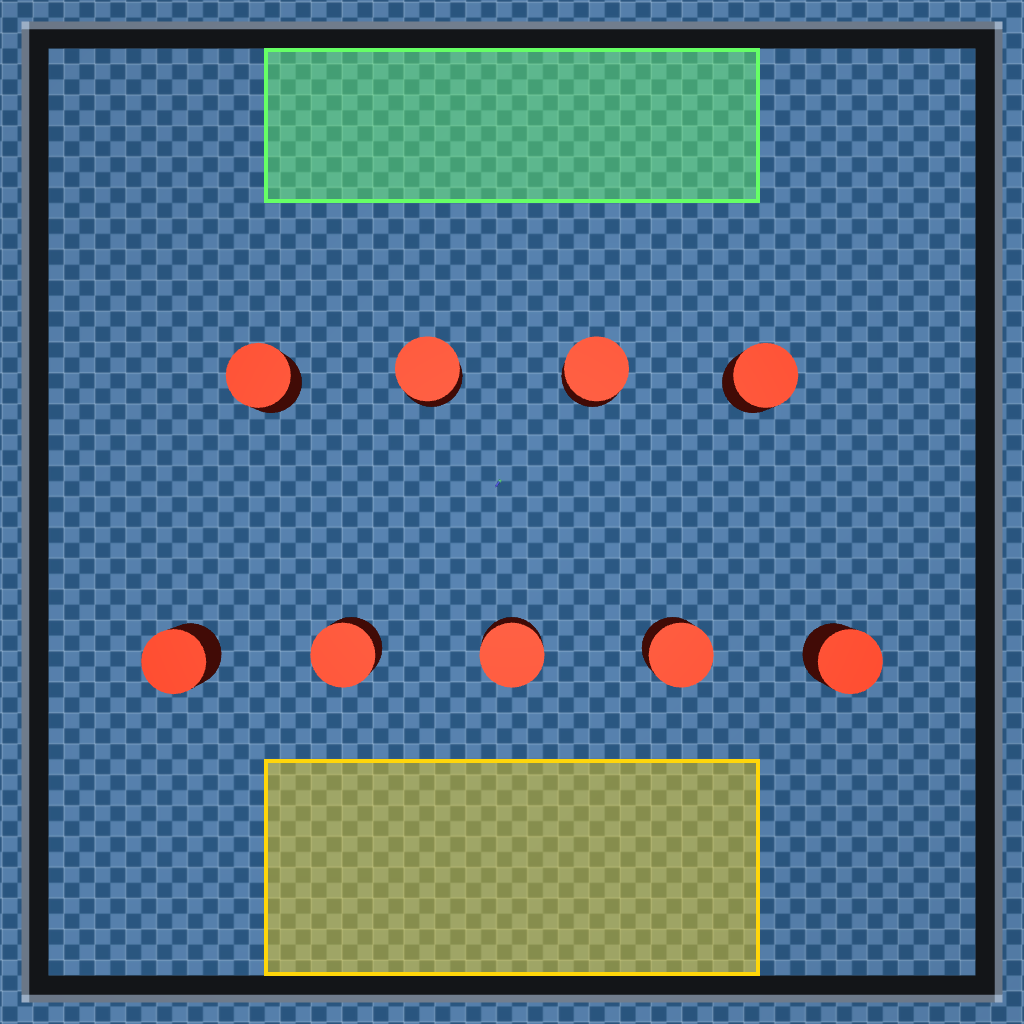}
    \caption{Low-Shift}
    \label{fig:map_low}
\end{subfigure}
\hfill
\begin{subfigure}[t]{0.32\textwidth}
    \centering
    \includegraphics[width=\linewidth]{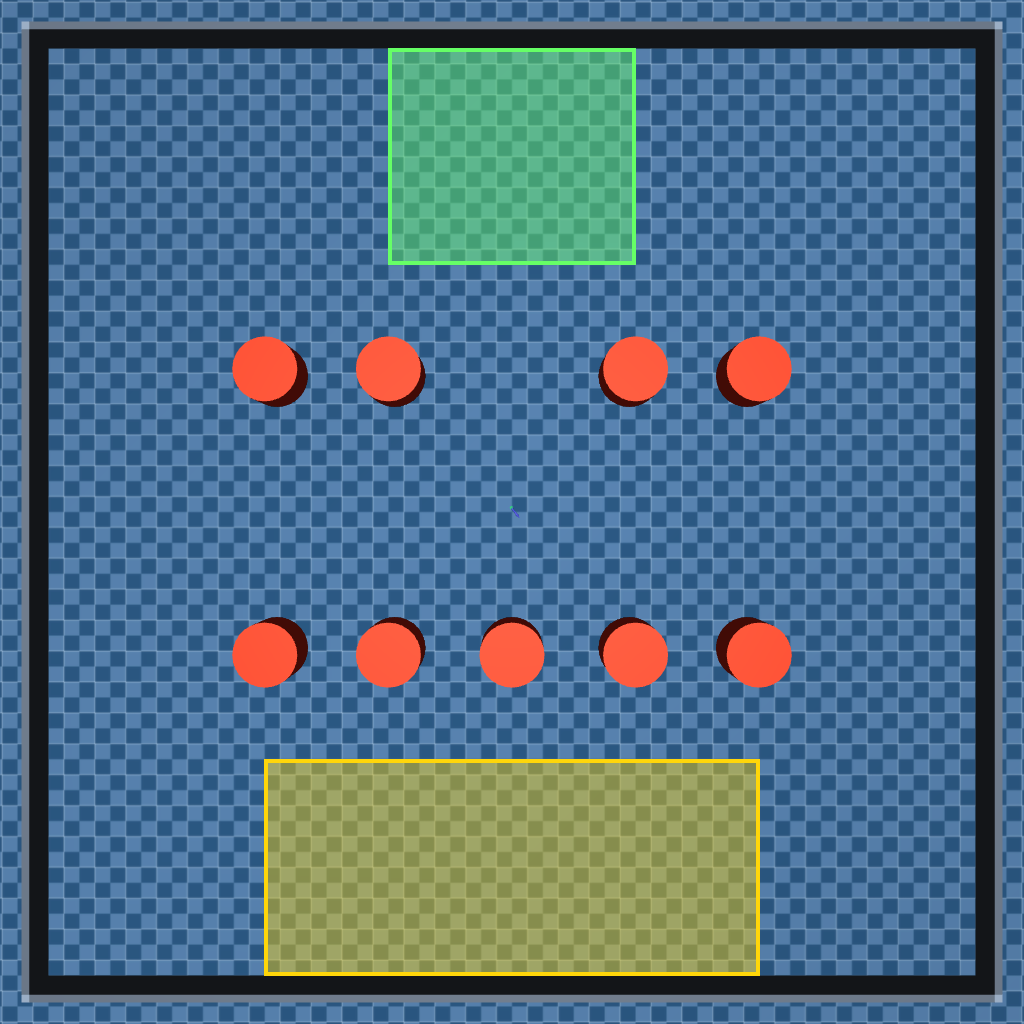}
    \caption{Medium-Shift}
    \label{fig:map_medium}
\end{subfigure}
\hfill
\begin{subfigure}[t]{0.32\textwidth}
    \centering
    \includegraphics[width=\linewidth]{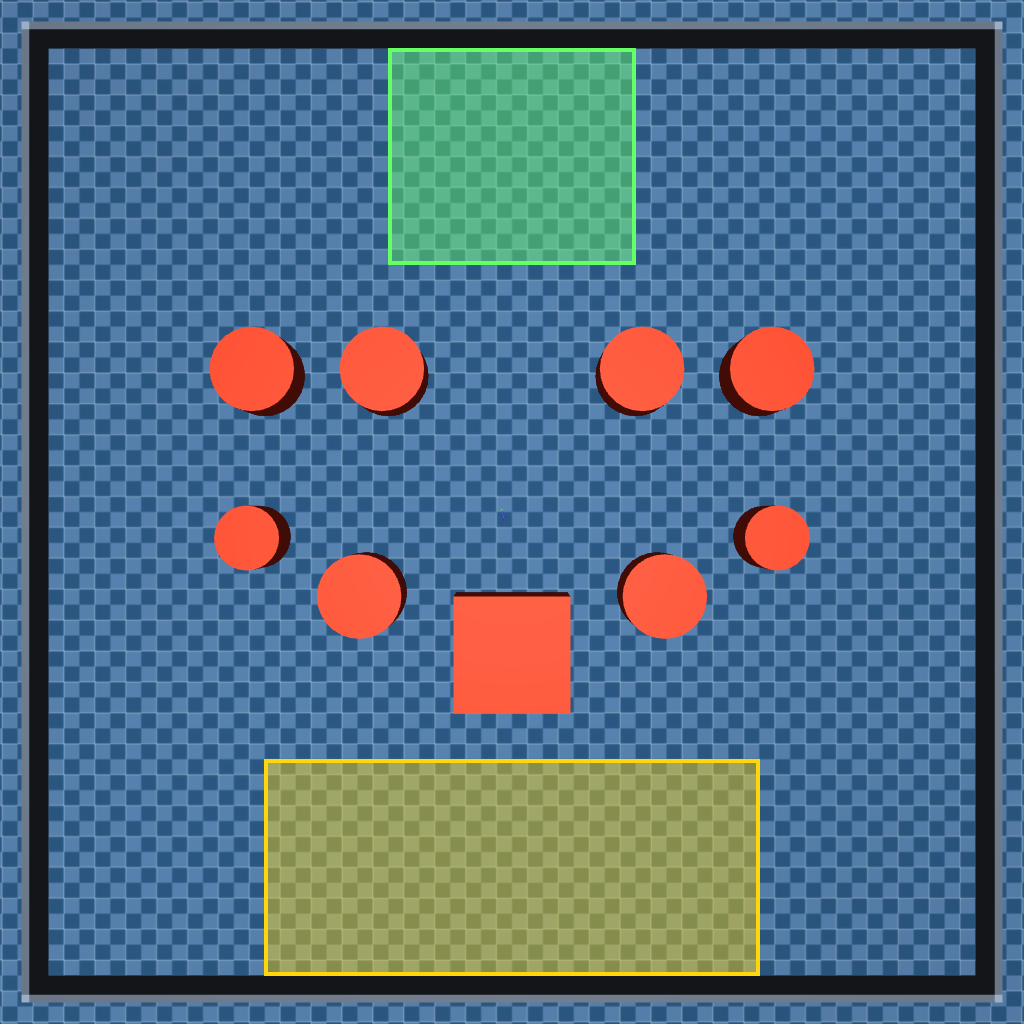}
    \caption{High-Shift}
    \label{fig:map_high}
\end{subfigure}

\caption{Static unseen environments used for distribution-shift evaluation. Green and yellow regions indicate the admissible initial and goal regions, respectively, while obstacles are shown in red. The environments introduce increasingly different obstacle configurations relative to the hand-crafted training maps.}
\label{fig:unseen_maps}
\end{figure*}

\subsection{Path Length and Navigation Time of Filter-Enabled Evaluations on Unseen Environments}
\label{app:pl_nt_unseen_active}
Table~\ref{tab:path_time_filter_active_unseen} reports path length and navigation time only over successful episodes for the filter-enabled unseen-environment evaluations.
Path lengths are broadly comparable across the SAC/CBF methods, while SAFTER generally requires more navigation time than the current-state baselines and less than bCBF--LSE--QP or CBF-RL.
Thus the success-rate gains under geometric shift should not be interpreted as uniformly shorter or faster successful trajectories.
Because PL and NT are conditioned on success, they should be interpreted together with the success rates in Table~\ref{tab:eval_filter_active}.

\begin{table}[h]
\centering
\caption{%
Navigation efficiency on the unseen static environments with the runtime safety filters enabled.
Path Length (PL) and Navigation Time (NT) are computed over successful episodes and reported as mean $\pm$ standard deviation. Dashes indicate that no successful episodes were available for computing PL and NT.
}
\label{tab:path_time_filter_active_unseen}

\resizebox{\columnwidth}{!}{%
\begin{tabular}{lcccccc}
\toprule
&
\multicolumn{2}{c}{\textbf{Low-Shift}}
&
\multicolumn{2}{c}{\textbf{Medium-Shift}}
&
\multicolumn{2}{c}{\textbf{High-Shift}}
\\
\cmidrule(lr){2-3}
\cmidrule(lr){4-5}
\cmidrule(lr){6-7}

Method
& \textbf{PL (m)} $\downarrow$ & \textbf{NT (s)} $\downarrow$
& \textbf{PL (m)} $\downarrow$ & \textbf{NT (s)} $\downarrow$
& \textbf{PL (m)} $\downarrow$ & \textbf{NT (s)} $\downarrow$
\\
\midrule

SEA--Nav\textsuperscript{1}
& $11.57{\pm}0.17$ & $13.70{\pm}0.17$
& $10.95{\pm}0.30$ & $14.21{\pm}0.27$
& $--$ & $--$ \\

CBF--RL\textsuperscript{2}
& $14.86{\pm}0.78$ & $32.02{\pm}4.02$
& $14.93{\pm}0.76$ & $34.93{\pm}4.3$
& $15.72{\pm}1.19$ & $38.99{\pm}8.45$ \\

\midrule

Vanilla SAC\textsuperscript{3}
& $12.76{\pm}0.41$ & $26.75{\pm}1.54$
& $12.38{\pm}0.68$ & $26.31{\pm}2.13$
& $12.20{\pm}0.70$ & $25.97{\pm}1.77$ \\

CBF--QP\textsuperscript{4}
& $12.05{\pm}0.38$ & $25.08{\pm}1.34$
& $11.50{\pm}0.45$ & $23.77{\pm}1.55$
& $11.82{\pm}1.11$ & $25.42{\pm}4.02$ \\

CBF--CF\textsuperscript{4}
& $12.26{\pm}0.36$ & $26.85{\pm}1.06$
& $11.72{\pm}0.42$ & $24.84{\pm}1.22$
& $11.72{\pm}0.35$ & $25.93{\pm}2.06$ \\

bCBF--LSE--QP\textsuperscript{5}
& $13.06{\pm}0.37$ & $32.74{\pm}1.11$
& $12.20{\pm}0.45$ & $29.50{\pm}1.67$
& $12.33{\pm}0.54$ & $29.91{\pm}2.07$ \\

SAFTER\textsuperscript{5}
& $12.73{\pm}0.32$ & $30.66{\pm}1.18$
& $11.87{\pm}0.38$ & $28.10{\pm}1.40$
& $12.03{\pm}0.32$ & $29.43{\pm}1.34$ \\

\bottomrule
\end{tabular}%
}
{\scriptsize
    \textit{Refs.:}
    \textsuperscript{1}~\citep{chen2026sea};
    \textsuperscript{2}~\citep{yang2025cbf};
    \textsuperscript{3}~\citep{haarnoja2018soft};
    \textsuperscript{4}~\citep{rabiee2024closed};
    \textsuperscript{5}~\textbf{Ours}.
    }
\vspace{-6mm}
\end{table}

\section{Theoretical Foundations and Proofs}
\label{app:theory}

\subsection{Direct Sampled bCBF Projection}
\label{app:sampled_bcbf_qp}
The bCBF--QP runtime reference is defined in the same augmented surrogate-control coordinates as SAFTER.
In particular, it uses the augmented state $z=[x^\top,u^\top]^\top$, surrogate input $\hat u$, prescribed surrogate backup $\hat u_{\mathrm b}$, backup flow $\varphi_{\mathrm b}^{z}$, lifted current-time margins $\psi_i$, positive-time backup margins $h_i^{\mathrm b}(z,\tau_k)$, terminal margins $h_{B,j}^{\mathrm b}(z,T)$, and physical input-admissibility margins $I_d^\pm(z)$ introduced in Section~\ref{sec:proposed_method}.
Instead of composing these quantities with the hierarchical LSE certificate,
the direct baseline enforces their CBF conditions separately:
\begin{equation*}
\begin{aligned}
    \hat u_t^{\mathrm{bCBF\text{-}QP}}
    ={}&\arg\min_{\hat{u}\in\mathbb{R}^m}
        \frac12\|\hat u-\hat u_t^{\mathrm{nom}}\|_2^2\\
    \mathrm{s.t.}\quad&
        L_{\hat f}\psi_i(z_t)
        +L_{\hat g}\psi_i(z_t)\hat u
        +\alpha_0(\psi_i(z_t))\geq0,
        &&\forall i,\\
    &L_{\hat f}h_i^{\mathrm b}(z_t,\tau_k)
      +L_{\hat g}h_i^{\mathrm b}(z_t,\tau_k)\hat u
      +\alpha_{\mathrm b}(h_i^{\mathrm b}(z_t,\tau_k))\geq0,
        &&\forall i,k,\\
    &L_{\hat f}h_{B,j}^{\mathrm b}(z_t,T)
      +L_{\hat g}h_{B,j}^{\mathrm b}(z_t,T)\hat u
      +\alpha_B(h_{B,j}^{\mathrm b}(z_t,T))\geq0,
        &&\forall j, \\
    & L_{\hat{f}}I_d^+(z_t) + L_{\hat{g}}I_d^+(z_t)\hat{u}  + \alpha_I\!\left(I_d^+(z_t)\right) \ge 0, \quad d=1,\ldots,m, \\
     & L_{\hat{f}}I_d^-(z_t) + L_{\hat{g}}I_d^-(z_t)\hat{u} + \alpha_I\!\left(I_d^-(z_t)\right) \ge 0, \quad d=1,\ldots,m.
\end{aligned}
\end{equation*}
Here, $\alpha_0$, $\alpha_{\mathrm b}$, $\alpha_B$, and $\alpha_I$ denote the class-$\mathcal{K}$ functions associated with the respective CBF conditions.

The Lie derivatives of the future and terminal terms include the sensitivity
of $\varphi_{\mathrm b}^{z}$ to its initial augmented state, as derived in
Appendix~\ref{app:backup_sensitivity}.
At finite policy rate, the experimental bCBF--QP baseline uses the same command-state dynamics, command limits, backup horizon, and terminal set as SAFTER.
Its synthesis difference is that the sampled safety conditions remain separate QP constraints, whereas SAFTER first composes them into a single affine condition.
Consequently, bCBF--QP and SAFTER use the same command-state augmentation, backup maneuver, horizon samples, terminal set, and physical command limits.
Their control-synthesis difference is that bCBF--QP retains one affine constraint per constituent, whereas SAFTER first composes the constituents and projects onto a single composite half-space.
Sampling the backup horizon does not, by itself, certify safety between adjacent $\tau_k$.
The finite-rate and inter-sample limitations are discussed in Appendix~\ref{app:guarantee_scope}.

\subsection{LSE Properties}
\label{app:lse_properties}
For the set-based smooth minimum in \eqref{eq:lse_main}, enumerate
$\mathcal S=\{c_\ell\}_{\ell=1}^{M}$. Its gradient is
\begin{equation}
    \nabla\softmin_{\kappa}(\mathcal S)
    =\sum_{\ell=1}^{M}w_\ell\nabla c_\ell,
    \qquad
    w_\ell
    =\frac{e^{-\kappa c_\ell}}
           {\sum_{r=1}^{M}e^{-\kappa c_r}},
    \label{eq:app_lse_gradient}
\end{equation}
so $w_\ell>0$ and $\sum_\ell w_\ell=1$. Hence every constituent contributes
to the gradient of the single composite certificate.

\begin{proof}[Proof of the LSE bound in \eqref{eq:lse_main}]
Let $c_{\min}=\min_\ell c_\ell$. Since
\[
    e^{-\kappa c_{\min}}
    \leq\sum_{\ell=1}^{M}e^{-\kappa c_\ell}
    \leq M e^{-\kappa c_{\min}},
\]
taking $-\kappa^{-1}\log(\cdot)$ gives the stated bound. Differentiating the
same expression gives \eqref{eq:app_lse_gradient}.
\end{proof}

\subsection{Relative Degree and Backup-Flow Sensitivity}
\label{app:hocbf_relative_degree}
Because $h_i$ depends on $x$ but not on the command state $u$,
$L_{\hat g}h_i=0$ under \eqref{eq:augmented_dynamics}. For
\[
    \psi_i(z)=L_fh_i(x)+L_gh_i(x)u+k_0h_i(x),
\]
we have
\[
    \frac{\partial\psi_i}{\partial u}=L_gh_i(x),
    \qquad
    L_{\hat g}\psi_i=L_gh_i(x)A.
\]
Thus $\psi_i$ has relative degree one with respect to $\hat u$ wherever
$L_gh_i(x)A\neq0$. Moreover, $\psi_i=\dot h_i+k_0h_i\geq0$ together with
$h_i(x(0))\geq0$ implies $h_i(x(t))\geq0$ for the continuous augmented model.

\subsubsection{Backup-flow sensitivity}
\label{app:backup_sensitivity}
Let
\[
    \hat f_{\mathrm b}(z):=\hat f(z)+\hat g(z)\hat u_{\mathrm b}(z),
    \qquad
    z_{\mathrm b}(\tau)=\varphi_{\mathrm b}^{z}(\tau;z_0).
\]
The sensitivity
\[
    J_{\mathrm b}(\tau,z_0)
    :=\frac{\partial\varphi_{\mathrm b}^{z}(\tau;z_0)}{\partial z_0}
\]
satisfies
\[
    \dot J_{\mathrm b}
    =\frac{\partial\hat f_{\mathrm b}}{\partial z}
      (z_{\mathrm b}(\tau))J_{\mathrm b},
    \qquad
    J_{\mathrm b}(0,z_0)=I_{n+m}.
\]
For $E_x=[I_n\;0_{n\times m}]$,
\[
    \nabla_{z_0}h_i^{\mathrm b}(z_0,\tau)
    =J_{\mathrm b}(\tau,z_0)^\top E_x^\top
      \nabla_xh_i(E_xz_{\mathrm b}(\tau)).
\]
The terminal gradients follow by the same chain rule. These sensitivities are
what allow the future and terminal margins to contribute to
$L_{\hat g}h_{\mathrm{aug}}$.

\subsection{Closed-Form Projection: KKT Derivation}
\label{app:kkt_derivation}
Consider \eqref{eq:closed_form_qp} and define only for this derivation
\[
    \Omega_t:=\mathbf a(z_t)\hat u_t^{\mathrm{nom}}+b(z_t).
\]
With multiplier $\mu\geq0$, the Lagrangian is
\[
    \mathcal L(\hat u,\mu)
    =\frac12\|\hat u-\hat u_t^{\mathrm{nom}}\|_2^2
     -\mu\bigl(\mathbf a(z_t)\hat u+b(z_t)\bigr).
\]
Stationarity gives
$\hat u=\hat u_t^{\mathrm{nom}}+\mu\mathbf a(z_t)^\top$.
If $\|\mathbf a(z_t)\|_2>0$, complementary slackness yields
\[
    \mu^\star
    =\frac{\max\{0,-\Omega_t\}}{\|\mathbf a(z_t)\|_2^2},
\]
which recovers \eqref{eq:closed_form_solution}. If
$\mathbf a(z_t)=0$, the surrogate input has no local authority over the
certificate: $b(z_t)\geq0$ makes every $\hat u$ feasible, whereas
$b(z_t)<0$ is locally infeasible.

\subsection{Continuous Safety Assumptions and Proof}
\label{app:safety_proof}
\begin{assumption}[Continuous composite-CBF conditions]
\label{ass:continuous_safety_conditions}
On the domain where Theorem~\ref{thm:continuous_composite_safety} is invoked:
(i) the augmented dynamics, backup-flow terms, terminal functions, and
$h_{\mathrm{aug}}$ are continuously differentiable;
(ii) the augmented model represents the continuous high-level safety model;
(iii) the terminal backup set $\mathcal C_B^z$ is forward invariant under the
prescribed surrogate backup $\hat u_{\mathrm b}$ and satisfies
$E_x(\mathcal C_B^z)\subseteq\mathcal C_S$;
(iv) $h_i(x(0))\geq0$ for every primitive barrier and
$h_{\mathrm{aug}}(z(0))\geq0$; and
(v) the composite affine CBF condition is feasible whenever it is enforced.
\end{assumption}

\begin{proof}[Proof of Theorem~\ref{thm:continuous_composite_safety}]
The projected surrogate satisfies \eqref{eq:composite_affine_condition},
which is the standard CBF condition for $h_{\mathrm{aug}}$ under the augmented
dynamics. Continuous enforcement therefore makes
$\mathcal C_{\mathrm{aug}}=\{z:h_{\mathrm{aug}}(z)\geq0\}$ forward invariant.

Because the LSE smooth minimum under-approximates the ordinary minimum,
$h_{\mathrm{aug}}\geq0$ implies nonnegativity of every normalized constituent,
and each fixed positive normalization preserves sign. In particular,
$\psi_i=\dot h_i+k_0h_i\geq0$ and $h_i(x(0))\geq0$ imply
\[
    h_i(x(t))\geq h_i(x(0))e^{-k_0t}\geq0,
\]
so the reduced safety-model state remains in $\mathcal C_S$. The
positive-time backup margins are nonnegative at every selected $\tau_k$, and
the terminal margins place the predicted backup state in $\mathcal C_B^z$ at
$T$. By Assumption~\ref{ass:continuous_safety_conditions}(iii), continuing the
prescribed backup from that terminal state remains in $\mathcal C_B^z$ and
therefore in the original safe set after projection through $E_x$. Finally,
the input-admissibility constituents remain nonnegative, so the continuous
command state remains within its prescribed component-wise bounds. These are
the safety and sampled-horizon recoverability consequences stated in the
theorem. Inter-sample and finite-rate limitations are separated explicitly in
Appendix~\ref{app:guarantee_scope}.
\end{proof}

\section{Quadruped Navigation Specialization}
\label{app:quadruped_safety}

\subsection{Policy Observation, Action, and Safety-Filter Inputs}
\label{app:pomdp_details}

The high-level actor observation is
\begin{equation*}
    o_t=
    \begin{bmatrix}
        (\bar p_{G,t}^{B})^\top &
        \nu_t^\top &
        (u_{t-1}^{\mathrm{safe}})^\top &
        \bar\rho_t^\top &
        (u_{t-1}^{\mathrm{nom}})^\top &
        \iota_{t-1}
    \end{bmatrix}^\top
    \in\mathbb R^{N_h+12}.
\end{equation*}
Here, $p_{G,t}^{B}\in\mathbb R^2$ is the body-frame goal displacement and $\bar p_{G,t}^{B}=p_{G,t}^{B}/d_{\mathrm{norm}}$ is its normalized form with $d_{\mathrm{norm}}=7.5\,\mathrm m$ in the implementation.
The vector $\nu_t=[v_{x,t}^{\mathrm{meas}},v_{y,t}^{\mathrm{meas}},\omega_{z,t}^{\mathrm{meas}}]^\top\in\mathbb R^3$ contains measured body-frame translational and yaw velocities for the actor observation; it is not the augmented control-input state used by the safety-filter model.
The quantities $u_{t-1}^{\mathrm{safe}},u_{t-1}^{\mathrm{nom}}\in\mathbb R^3$ are the previous executed and nominal physical control inputs, respectively, and $\bar\rho_t\in[0,1]^{N_h}$ is the LiDAR scan normalized by the configured maximum range.
The flag $\iota_{t-1}\in\{0,1\}$ indicates whether the previous nominal control input is valid for the finite-difference derivative in Appendix~\ref{app:nominal_tracking}.
For the $360^\circ$ configuration, $N_h=120$ and therefore $o_t\in\mathbb R^{132}$.
The high-level actor does not observe the full quadruped proprioceptive state because full-body stabilization and low-level tracking are handled by the frozen locomotion policy.

The normalized actor output $\tilde u_t^{\mathrm{nom}}\in[-1,1]^3$
is mapped to the physical nominal command according to
\begin{equation*}
    u_t^{\mathrm{nom}}
    =
    c_u+S_u\tilde u_t^{\mathrm{nom}},
    \qquad
    c_u
    =
    \frac{u_{\max}+u_{\min}}{2},
    \qquad
    S_u
    =
    \operatorname{diag}\!\left(
        \frac{u_{\max}-u_{\min}}{2}
    \right).
\end{equation*}
For the experiments, we used $ u_{\min}=[-1,-1,-1]^\top$ and $u_{\max}=[2,1,1]^\top$.
At policy step $t$, define $u_t^-:=u_{t-1}^{\mathrm{safe}}$ as the applied command state entering the first-order augmentation. Thus the augmented filter state uses $z_t=[x_t^\top,(u_t^-)^\top]^\top$, while the measured velocity $\nu_t$ is used only by the actor observation.
The policy observation $o_t$ and safety-filter input $\xi_t$ therefore contain related but nonidentical information. In particular, the filter does not use the goal displacement or measured body velocity when evaluating the analytical safety projection; it uses the current LiDAR scan and the command history required by the local augmented model.

\subsection{LiDAR Safety Construction}
\label{app:lidar_model}

\subsubsection{Scan-Fixed Local Geometry}
\label{app:scan_fixed_geometry}
At each high-level safety-filter update, the generic SAFTER plant state $x$ is specialized to the reduced planar navigation state
\begin{equation*}
    x=
    \begin{bmatrix}
        p_x & p_y & \theta
    \end{bmatrix}^{\top},
\end{equation*}
where $p=[p_x,p_y]^\top$ is planar position and $\theta$ is yaw. The applied high-level command is $u=[v_x,v_y,\omega_z]^\top$, giving
\begin{equation*}
\dot p=R(\theta)\begin{bmatrix}v_x\\v_y\end{bmatrix},
\qquad
\dot\theta=\omega_z,
\qquad
\dot u=-Au+A\hat u,
\end{equation*}
with
\[
R(\theta)=\begin{bmatrix}
\cos\theta & -\sin\theta\\
\sin\theta & \cos\theta
\end{bmatrix},
\qquad
A=\operatorname{diag}(a_x,a_y,a_\omega)\succ0.
\]
Thus the generic augmented state $z=[x^\top,u^\top]^\top$ becomes $z=[p_x,p_y,\theta,v_x,v_y,\omega_z]^\top$ for the navigation certificate.

A temporary scan-fixed frame is initialized at every filter update with
$p(0)=0$ and $\theta(0)=0$. Each valid LiDAR return is converted into an
obstacle-boundary point
\begin{equation*}
q_i =
p_L^B +
\rho_i
\begin{bmatrix}
\cos\beta_i\\
\sin\beta_i
\end{bmatrix},
\end{equation*}
where $\rho_i$ and $\beta_i$ are the measured range and ray angle, and
$p_L^B$ is the LiDAR offset from the robot-body origin. The points $q_i$
are held fixed during one safety-filter computation.
As the robot evolves relative to this fixed scan, $q_i$ remains fixed while its body-frame coordinate is the derived quantity
\begin{equation}
\zeta_i(x)
=
R(\theta)^\top(q_i-p).
\label{eq:moving_lidar_point}
\end{equation}

\subsubsection{Elliptical Safety Barrier}

The robot footprint and safety clearance are conservatively represented by a body-aligned ellipse with longitudinal and lateral semi-axes $r_x$ and $r_y$. The semi-axes account for the physical robot dimensions, LiDAR offset, and configured safety clearance; for the verified Go2 configuration, $r_x=0.58945$ m and $r_y=0.39445$ m. Define
\begin{equation*}
    \Pi=\operatorname{diag}(r_x^{-2},r_y^{-2}).
\end{equation*}
For obstacle point $\zeta_i=[\zeta_{i,x},\zeta_{i,y}]^\top$ in the instantaneous body frame, the geometric barrier is
\begin{equation*}
    h_i(p,\theta)
    =
    \zeta_i^\top \Pi \zeta_i-1
    =
    \frac{\zeta_{i,x}^2}{r_x^2}+\frac{\zeta_{i,y}^2}{r_y^2}-1.
\end{equation*}
Thus, $h_i\geq0$ requires the observed obstacle point to lie outside the inflated body-aligned ellipse.
The elliptical footprint permits yaw to influence the barrier whenever $r_x\neq r_y$.
Let
\begin{equation*}
    J
    =
    \begin{bmatrix}
        0 & -1\\
        1 & 0
    \end{bmatrix}.
\end{equation*}
Differentiating \eqref{eq:moving_lidar_point} gives $\dot \zeta_i=-[v_x,v_y]^\top-\omega_zJ\zeta_i$ and hence
\begin{equation*}
    \dot h_i
    =
    -2\zeta_i^\top \Pi
    \begin{bmatrix}
        v_x\\v_y
    \end{bmatrix}
    -
    2\omega_z \zeta_i^\top \Pi J\zeta_i.
\end{equation*}
For the directly applied high-level control input $u=[v_x,v_y,\omega_z]^\top$, the corresponding CBF row is
\begin{equation*}
    L_gh_i
    =
    \begin{bmatrix}
        -2r_x^{-2}\zeta_{i,x}
        &
        -2r_y^{-2}\zeta_{i,y}
        &
        -2(r_y^{-2}-r_x^{-2})\zeta_{i,x}\zeta_{i,y}
    \end{bmatrix}.
\end{equation*}
The yaw coefficient is generally nonzero for an anisotropic ellipse and vanishes in the circular special case $r_x=r_y$.

\subsubsection{Current-Time HOCBF Condition}

Under the augmented control-input dynamics, the primitive barrier $h_i$ has no first-order dependence on the surrogate control input $\hat u$.
The current-time lifted constituent is therefore
\begin{equation*}
    \psi_i(z)
    =
    \dot h_i(z)+k_0h_i(z),
    \qquad
    k_0>0.
\end{equation*}
At the beginning of each scan-fixed computation, $p=0$ and $\theta=0$, so $\zeta_i=q_i$ in the temporary local frame.
To avoid overloading the plant input matrix $g(x)$ or the KKT multiplier $\mu$, define the two point-gradient components
\begin{align*}
    \Gamma_{i,x}^{\psi}
    &:=
    -2r_x^{-2}v_x
    -2\omega_z(r_y^{-2}-r_x^{-2})\zeta_{i,y}
    +2k_0r_x^{-2}\zeta_{i,x},
    \\
    \Gamma_{i,y}^{\psi}
    &:=
    -2r_y^{-2}v_y
    -2\omega_z(r_y^{-2}-r_x^{-2})\zeta_{i,x}
    +2k_0r_y^{-2}\zeta_{i,y}.
\end{align*}
The complete gradient with respect to $z=[p_x,p_y,\theta,v_x,v_y,\omega_z]^\top$ is
\begin{equation*}
\nabla_z\psi_i
=
\begin{bmatrix}
    -\Gamma_{i,x}^{\psi}\\
    -\Gamma_{i,y}^{\psi}\\
    \Gamma_{i,x}^{\psi}\zeta_{i,y}-\Gamma_{i,y}^{\psi}\zeta_{i,x}\\
    -2r_x^{-2}\zeta_{i,x}\\
    -2r_y^{-2}\zeta_{i,y}\\
    -2(r_y^{-2}-r_x^{-2})\zeta_{i,x}\zeta_{i,y}
\end{bmatrix}.
\end{equation*}
Its last three entries give $\partial\psi_i/\partial u=L_gh_i$, so multiplication by the control-state input block $A$ recovers $L_{\hat g}\psi_i=L_gh_iA$ as stated in Section~\ref{sec:augmentation_lifting}.
The first three entries account for the dependence of the scan-fixed point coordinate on robot translation and yaw.
Together, these expressions provide the code-aligned current-time gradient used in the hierarchical certificate.

\subsubsection{Backup Rollout and Positive-Time Barrier}

The quadruped specialization uses the zero surrogate backup $\hat u_{\mathrm b}=0$.
Along this backup flow, the applied control-input state decays as
\begin{equation*}
    u(\tau)
    =
    e^{-A\tau}u_0.
\end{equation*}
The corresponding yaw trajectory and planar position are
\begin{equation*}
    \theta_{\mathrm b}(\tau)
    =
    \theta_0
    +
    \omega_{z,0}
    \frac{1-e^{-a_\omega\tau}}{a_\omega},
\end{equation*}
\begin{equation*}
    p_{\mathrm b}(\tau)
    =
    p_0
    +
    \int_0^\tau
    R(\theta_{\mathrm b}(\sigma))
    \begin{bmatrix}
        v_{x,0}e^{-a_x\sigma}\\
        v_{y,0}e^{-a_y\sigma}
    \end{bmatrix}
    d\sigma.
\end{equation*}
These equations define the exact continuous augmented backup flow $z_{\mathrm b}(\tau)=\varphi_{\mathrm b}^{z}(\tau;z_0)$ used in the theoretical construction. The implementation evaluates the planar integral with fixed eight-point Gauss--Legendre quadrature and differentiates that fixed quadrature map analytically. Thus the implemented rollout is a smooth numerical approximation of the exact flow; it does not use sequential Euler integration or finite-difference flow sensitivities. All requested horizon states and their implementation-level sensitivities are therefore evaluated as smooth functions of the same initial augmented state.

The implementation used a total of 21 rollout nodes, including the current-time node at $\tau_0=0$, along the trajectory generated by the backup controller.
Because the raw zero-time geometric barrier is replaced by the lifted constituent $\psi_i$, the LSE receives $K=20$ positive-time samples,
\begin{equation*}
    \tau_k
    =
    \frac{kT}{20}
    =
    0.2k\,\mathrm s,
    \qquad
    k=1,\ldots,20,
    \qquad
    T=4.0\,\mathrm s.
\end{equation*}
For each positive-time sample, the backup safety value is
\begin{equation*}
    h_i^{\mathrm b}(z_0,\tau_k)
    =
    h_i\!\left(
        p_{\mathrm b}(\tau_k),
        \theta_{\mathrm b}(\tau_k)
    \right).
\end{equation*}
Let
\begin{equation*}
    J_{\mathrm b}(\tau,z_0)
    =
    \frac{\partial\varphi_{\mathrm b}^{z}(\tau;z_0)}{\partial z_0}
\end{equation*}
denote the augmented backup-flow sensitivity.
At a future state, the primitive barrier gradient is
\begin{equation*}
\nabla_z h_i
=
\begin{bmatrix}
    -\left(R(\theta)(2\Pi\zeta_i)\right)_1\\
    -\left(R(\theta)(2\Pi\zeta_i)\right)_2\\
    (2\Pi\zeta_i)_x \zeta_{i,y}-(2\Pi\zeta_i)_y \zeta_{i,x}\\
    0\\0\\0
\end{bmatrix}.
\end{equation*}
The flow-composed gradient with respect to the current augmented state is
\begin{equation*}
    \nabla_{z_0}h_i^{\mathrm b}(z_0,\tau_k)
    =
    J_{\mathrm b}(\tau_k,z_0)^\top
    \nabla_z h_i\!\left(z_{\mathrm b}(\tau_k)\right).
\end{equation*}
This chain rule is the quantity used when the positive-time barriers contribute to the LSE gradient and its Lie derivatives.

\subsubsection{Forward-Invariant Elliptical Terminal Backup Set}
\label{app:terminal_set}

The terminal construction uses a fixed enlarged ellipse together with a separate residual-motion condition, rather than embedding a smoothed stopping distance directly into the ellipse semi-axes.
Define
\[
R_x=r_x+d_{\mathrm b},
\qquad
R_y=r_y+d_{\mathrm b},
\qquad
S_{B}=\operatorname{diag}(R_x,R_y),
\]
where $d_{\mathrm b}>0$ is the additional terminal margin.
For terminal obstacle coordinate $\zeta_i=[\zeta_{i,x},\zeta_{i,y}]^\top$, define
\[
E_i
=
\zeta_i^\top S_{B}^{-2}\zeta_i
=
\frac{\zeta_{i,x}^2}{R_x^2}
+
\frac{\zeta_{i,y}^2}{R_y^2},
\qquad
\ell_i
=
\frac12\log E_i.
\]
Thus, $\ell_i\geq0$ is equivalent to $E_i\geq1$ and selects the exterior of the fixed terminal ellipse.
The logarithmic coordinate is dimensionless and is well defined on the evaluation domain $E_i>0$.

To quantify residual motion under the zero-surrogate backup, define
\[
c_\theta
=
\frac12
\left|
\frac{R_x}{R_y}
-
\frac{R_y}{R_x}
\right|,
\qquad
V_{\mathrm b}
=
\frac{v_x^2}{R_x^2}
+
\frac{v_y^2}{R_y^2}
+
c_\theta^2\omega_z^2,
\]
and let $\lambda_A:=\min(a_x,a_y,a_\omega)$ denote the lowest value of the diagonal matrix $A$.
The quantity $V_{\mathrm b}$ has units of inverse time squared and measures normalized residual translation together with a bound on the yaw contribution induced by ellipse anisotropy.
Under $\hat u_{\mathrm b}=0$, the components satisfy $\dot v_x=-a_xv_x$, $\dot v_y=-a_yv_y$, and $\dot\omega_z=-a_\omega\omega_z$.
Consequently,
\begin{equation}
    \dot V_{\mathrm b}
    \leq
    -2\lambda_A\,V_{\mathrm b}.
    \label{eq:terminal_motion_decay}
\end{equation}

For each valid LiDAR point, we define two terminal backup-set constituents: $h_{B,i}^{\mathrm{geom}}$ for geometric clearance and $h_{B,i}^{\mathrm{rec}}$ for recoverability under the residual commanded motion.
\begin{equation*}
    h_{B,i}^{\mathrm{geom}}
    :=
    \ell_i,
    \qquad
    h_{B,i}^{\mathrm{rec}}
    :=
    \ell_i^2
    -
    \frac{2V_{\mathrm b}}{\lambda_A^2}.
\end{equation*}
The per-ray set is $\mathcal C_{B,i}^z:=\{z:h_{B,i}^{\mathrm{geom}}(z)\geq0,\ h_{B,i}^{\mathrm{rec}}(z)\geq0\}$ and the complete augmented terminal set is $\mathcal C_B^z:=\bigcap_i\mathcal C_{B,i}^z$. The generic terminal index $j$ in Section~\ref{sec:predictive_backup} therefore ranges over these two terminal constituents for all LiDAR rays.
Neither constituent is claimed to be forward invariant independently; the invariant object is their intersection.
Equivalently, inside the geometric branch the recoverability condition requires $\ell_i\geq\sqrt{2V_{\mathrm b}}/\lambda_A$.
At zero residual motion, the set reduces to the exterior of the fixed enlarged terminal ellipse.

\begin{proposition}[Forward invariance of the terminal backup set]
\label{prop:terminal_invariance}
For the fixed-scan local model in Appendix~\ref{app:scan_fixed_geometry}, $\mathcal C_B^z$ is forward invariant under the continuous zero-surrogate backup controller $\hat u_{\mathrm b}=0$.
Moreover, $E_x(\mathcal C_B^z)\subseteq\mathcal C_S$ because $R_x>r_x$ and $R_y>r_y$.
\end{proposition}

\begin{proof}
Let
\[
w_i=S_{B}^{-1}\zeta_i,
\qquad
\bar v_{B}
=
S_{B}^{-1}
\begin{bmatrix}
v_x\\v_y
\end{bmatrix},
\qquad
M_{B}
=
S_{B}^{-1}JS_{B}.
\]
Because $E_i=\|w_i\|_2^2$, we have $\ell_i=\log\|w_i\|_2$.
Using $\dot \zeta_i=-[v_x,v_y]^\top-\omega_zJ\zeta_i$ gives $\dot w_i=-\bar v_{B}-\omega_zM_{B}w_i$.
The symmetric part of $M_{B}$ has spectral norm $c_\theta$, so $|w_i^\top M_{B}w_i|\leq c_\theta\|w_i\|_2^2$.
Therefore, on the geometric branch $\ell_i\geq0$, where $\|w_i\|_2\geq1$,
\begin{align}
\dot\ell_i
&=
-\frac{w_i^\top\bar v_{B}}{\|w_i\|_2^2}
-
\omega_z
\frac{w_i^\top M_{B}w_i}{\|w_i\|_2^2}
\\
&\geq
-\|\bar v_{B}\|_2
-c_\theta|\omega_z|
\\
&\geq
-\sqrt{2V_{\mathrm b}}.
\label{eq:terminal_log_bound}
\end{align}
Consider a boundary point with $h_{B,i}^{\mathrm{rec}}=0$ and $h_{B,i}^{\mathrm{geom}}\geq0$.
Then $\ell_i=\sqrt{2V_{\mathrm b}}/\lambda_A$, and \eqref{eq:terminal_motion_decay}--\eqref{eq:terminal_log_bound} give
\begin{align*}
\dot h_{B,i}^{\mathrm{rec}}
&=
2\ell_i\dot\ell_i
-
\frac{2}{\lambda_A^2}\dot V_{\mathrm b}
\\
&\geq
-2\ell_i\sqrt{2V_{\mathrm b}}
+
\frac{4V_{\mathrm b}}{\lambda_A}
=0.
\end{align*}
Hence, the backup vector field does not point outward through the recoverability boundary.
If instead $h_{B,i}^{\mathrm{geom}}=\ell_i=0$ while $h_{B,i}^{\mathrm{rec}}\geq0$, the recoverability condition implies $V_{\mathrm b}=0$ and therefore $v_x=v_y=\omega_z=0$.
Consequently $\dot\ell_i=0$, so the vector field is tangent to the geometric boundary.
The Nagumo boundary condition holds on every active boundary of $\mathcal C_{B,i}^z$, and the intersection over all valid fixed-scan LiDAR points is forward invariant under the same backup vector field.
Finally, the terminal ellipse strictly contains the ordinary safety ellipse because $R_x=r_x+d_{\mathrm b}>r_x$ and $R_y=r_y+d_{\mathrm b}>r_y$; being outside the terminal ellipse therefore implies being outside the ordinary ellipse, yielding $E_x(\mathcal C_B^z)\subseteq\mathcal C_S$.
\end{proof}

For the experimental geometry, $r_x=0.58945\,\mathrm m$, $r_y=0.39445\,\mathrm m$, and $d_{\mathrm b}=0.60\,\mathrm m$.
Therefore $R_x=1.18945\,\mathrm m$, $R_y=0.99445\,\mathrm m$, and $c_\theta=0.180015$ to the reported precision.
With $A=\operatorname{diag}(8,8,8)\,\mathrm{s}^{-1}$, the residual-motion decay rate is $\lambda_A=8\,\mathrm{s}^{-1}$.
These values match the implementation summary in Appendix~\ref{app:filter_parameters}. Proposition~\ref{prop:terminal_invariance} verifies the terminal-set condition in Assumption~\ref{ass:continuous_safety_conditions}(iii) for this specialization.

For implementation, define $\Pi_{B}:=S_{B}^{-2}$, so $E_i=\zeta_i^\top \Pi_{B}\zeta_i$.
The terminal-state gradient of the logarithmic clearance can then be written without introducing an additional helper symbol as
\begin{equation*}
\nabla_z\ell_i
=
\begin{bmatrix}
-\left(R(\theta)\dfrac{\Pi_{B}\zeta_i}{E_i}\right)_1\\[1mm]
-\left(R(\theta)\dfrac{\Pi_{B}\zeta_i}{E_i}\right)_2\\[1mm]
\left(\dfrac{\Pi_{B}\zeta_i}{E_i}\right)_x \zeta_{i,y}
-
\left(\dfrac{\Pi_{B}\zeta_i}{E_i}\right)_y \zeta_{i,x}\\[1mm]
0\\0\\0
\end{bmatrix}.
\end{equation*}
The residual-motion gradient is
\begin{equation*}
\nabla_zV_{\mathrm b}
=
\begin{bmatrix}
0\\0\\0\\
2v_x/R_x^2\\
2v_y/R_y^2\\
2c_\theta^2\omega_z
\end{bmatrix}.
\end{equation*}
Hence $\nabla_z h_{B,i}^{\mathrm{geom}}=\nabla_z\ell_i$ and
\begin{equation*}
    \nabla_z h_{B,i}^{\mathrm{rec}}
    =
    2\ell_i\nabla_z\ell_i
    -
    \frac{2}{\lambda_A^2}\nabla_zV_{\mathrm b}.
\end{equation*}
The gradients with respect to the current augmented state are obtained through the same terminal backup-flow sensitivity used for the rollout constituents,
\begin{equation*}
    \nabla_{z_0}h_{B,i}^{(\cdot)}(z_0,T)
    =
    J_{\mathrm b}(T,z_0)^\top
    \nabla_z h_{B,i}^{(\cdot)}(z_{\mathrm b}(T)).
\end{equation*}
This preserves one common sensitivity pipeline for both terminal constituents and introduces no additional finite-difference approximation.

\subsubsection{Implementation-Specific LSE Scaling}
\label{app:lse_scaling}

The implementation applies only fixed, code-defined scaling before hierarchical LSE aggregation.
The current-time HOCBF constituent has units of inverse time and is normalized as
\begin{equation*}
    \bar\psi_i
    =
    \frac{\psi_i}{s_\psi},
    \qquad
    s_\psi=1.0\,\mathrm{s}^{-1}.
\end{equation*}
The positive-time geometric barriers $h_i^{\mathrm b}(z,\tau_k)$ are dimensionless and are inserted directly, so $\bar h_i^{\mathrm b}=h_i^{\mathrm b}$ with no separate rollout-scale hyperparameter.
The two terminal constituents are likewise dimensionless and use the fixed scales $s_{B,\mathrm{geom}}=s_{B,\mathrm{rec}}=1$, so their numerical values are unchanged by normalization.
Thus, no independent tunable rollout or terminal normalization is hidden from the parameter table.

The control-input margins are made dimensionless using fixed positive per-axis reference scales.
Each control-input component is defined as $s_{I,d}:= \max\!\left(|u_{\min,d}|,|u_{\max,d}|\right)> 0$.
The normalized input constituents are
\begin{equation*}
    \bar I_d^+
    =
    \frac{u_{\max,d}-u_d}{s_{I,d}},
    \qquad
    \bar I_d^-
    =
    \frac{u_d-u_{\min,d}}{s_{I,d}}.
\end{equation*}
The input bounds used in this work are $u_{\min}=[-1,-1,-1]^\top$ and $u_{\max}=[2,1,1]^\top$, which gives $s_I=[2,1,1]^\top$.
Accordingly, the longitudinal, lateral, and yaw-rate margins are normalized by
$2~\mathrm{m\,s^{-1}}$, $1~\mathrm{m\,s^{-1}}$, and $1~\mathrm{rad\,s^{-1}}$, respectively.
Because each scale is strictly positive, the normalization preserves the sign of the corresponding input-admissibility constituent while only affecting its numerical weighting within the LSE aggregation.

\subsubsection{Safety-Filter and Implementation Parameters}
\label{app:filter_parameters}

Table~\ref{tab:safety_filter_parameters} summarizes the principal parameters used by the high-level safety-filter configuration.
The high-level policy and safety filter operate at $10\,\mathrm{Hz}$, while the frozen low-level locomotion controller operates at $50\,\mathrm{Hz}$.
The backup flow uses the zero-surrogate controller with the same first-order control-input dynamics as the safety filter and evaluates its planar displacement with fixed eight-point Gauss--Legendre quadrature.
The rollout-count notation distinguishes the 21 total nodes stored by the implementation parameter from the 20 strictly positive-time samples that actually enter the LSE certificate.

\begin{table}[h]
\centering
\caption{Safety-filter and local-navigation parameters used by the verified implementation.}
\label{tab:safety_filter_parameters}
\small
\begin{tabular}{lll}
\toprule
Parameter & Symbol & Value \\
\midrule
\multicolumn{3}{l}{\textit{Control and sensing}}\\
LiDAR update rate & -- & $10\,\mathrm{Hz}$ \\
Low-level control rate & -- & $50\,\mathrm{Hz}$ \\
LiDAR field of view & -- & $360^\circ$ \\
Number of LiDAR rays & $N_h$ & $120$ \\
LiDAR range limits & $\rho_{\min},\rho_{\max}$ & $0.05,\ 5.0\,\mathrm{m}$ \\
LiDAR body offset & $p_L^B$ & $[0.28945,0]^\top\,\mathrm m$ \\
Robot length & $L$ & $0.70\,\mathrm{m}$ \\
Robot width & $w_{\mathrm{rob}}$ & $0.31\,\mathrm{m}$ \\
Configured safety distance & $d_{\mathrm{safe}}$ & $0.30\,\mathrm{m}$ \\
Effective footprint clearance & $d_{\mathrm{eff}}$ & $0.23945\,\mathrm{m}$ \\
Longitudinal CBF semi-axis & $r_x$ & $0.58945\,\mathrm{m}$ \\
Lateral CBF semi-axis & $r_y$ & $0.39445\,\mathrm{m}$ \\
\midrule
\multicolumn{3}{l}{\textit{Control-input augmentation}}\\
Control-dynamics gain & $A$ & $\operatorname{diag}(8,8,8)\,\mathrm{s}^{-1}$ \\
Nominal tracking gain & $\Sigma$ & $\operatorname{diag}(8,8,8)\,\mathrm{s}^{-1}$ \\
Current HOCBF gain & $k_0$ & $2.0\,\mathrm{s}^{-1}$ \\
Composite CBF gain & $\alpha_{\mathrm{aug}}$ & $0.8\,\mathrm{s}^{-1}$ \\
Longitudinal speed command limits & $v_x$ & $[-1.0,2.0]\,\mathrm{m\,s^{-1}}$ \\
Lateral speed command limits & $v_y$ & $[-1.0,1.0]\,\mathrm{m\,s^{-1}}$ \\
Yaw-rate command limits & $\omega_z$ & $[-1.0,1.0]\,\mathrm{rad\,s^{-1}}$ \\
\midrule
\multicolumn{3}{l}{\textit{Backup prediction and terminal set}}\\
Backup surrogate input & $\hat u_{\mathrm b}$ & $0$ \\
Backup horizon & $T$ & $4.0\,\mathrm{s}$ \\
Total rollout nodes & $N_{\mathrm b}$ & $21$ (includes $\tau_0=0$) \\
Positive-time LSE samples & $K$ & $20$ \\
Positive sample times & $\tau_k$ & $0.2k\,\mathrm{s},\ k=1{:}20$ \\
Fixed terminal margin & $d_{\mathrm b}$ & $0.60\,\mathrm{m}$ \\
Terminal longitudinal semi-axis & $R_x$ & $1.18945\,\mathrm{m}$ \\
Terminal lateral semi-axis & $R_y$ & $0.99445\,\mathrm{m}$ \\
Yaw anisotropy coefficient & $c_\theta$ & $0.180015$ \\
Residual-motion decay rate & $\lambda_A$ & $8.0\,\mathrm{s}^{-1}$ \\
\midrule
\multicolumn{3}{l}{\textit{LSE aggregation}}\\
Recoverability LSE sharpness & $\kappa_{B}$ & $50$ \\
Input LSE sharpness & $\kappa_{U}$ & $20$ \\
Outer LSE sharpness & $\kappa_{O}$ & $20$ \\
HOCBF normalization scale & $s_\psi$ & $1.0\,\mathrm{s}^{-1}$ \\
Terminal geometric scale & $s_{B,\mathrm{geom}}$ & $1.0$ \\
Terminal recoverability scale & $s_{B,\mathrm{rec}}$ & $1.0$ \\
\bottomrule
\end{tabular}
\vspace{-5mm}
\end{table}

\section{Finite-Rate Realization and Guarantee Scope}
\label{app:finite_rate_realization}

\subsection{Nominal Surrogate Tracking}
\label{app:nominal_tracking}
Let $u^{\mathrm{nom}}(t)$ denote the nominal physical command. Substituting
\eqref{eq:nominal_surrogate} into the augmented command dynamics gives
\[
    \dot u=\dot u^{\mathrm{nom}}+\Sigma(u^{\mathrm{nom}}-u),
    \qquad
    \dot e_u=-\Sigma e_u,
    \quad e_u:=u-u^{\mathrm{nom}}.
\]
Hence the applied command exponentially tracks the nominal command whenever
the safety projection is inactive. In the $10\,\mathrm{Hz}$ implementation,
\[
    \dot u_t^{\mathrm{nom}}
    \approx
    \iota_{t-1}\frac{u_t^{\mathrm{nom}}-u_{t-1}^{\mathrm{nom}}}{\Delta t},
\]
where $\iota_{t-1}=0$ when no valid previous nominal command is available and
$1$ otherwise.

\subsection{Finite-Rate Experimental Realization}
\label{app:finite_rate_exp}

The reported experiments evaluate the safety filter at the high-level control interval $\Delta t=0.1\,\mathrm{s}$.
The analytical projection in \eqref{eq:closed_form_solution} is first applied in surrogate-control space to obtain $\hat{u}_t^\star$.
During each high-level interval, $\hat{u}_t^\star$ is held constant under the first-order command dynamics
\[
    \dot{u}=-Au+A\hat{u}_t^\star .
\]
Because $A$ is diagonal, the corresponding command-state update is evaluated exactly as
\begin{equation}
    u_{t}^{\mathrm{pre}}
    =
    \Lambda u_t^-
    +
    (I-\Lambda)\hat{u}_t^\star,
    \qquad
    \Lambda:=e^{-A\Delta t}.
    \label{eq:finite_rate_command_update}
\end{equation}

The physical command executed by the environment is then obtained through the component-wise command limits,
\begin{equation}
    u_{t}^{\mathrm{safe}}
    =
    \operatorname{clip}
    \left(
        u_{t}^{\mathrm{pre}},
        u_{\min},
        u_{\max}
    \right).
    \label{eq:finite_rate_command_clip}
\end{equation}
At the next filter evaluation, this executed command is used as the current control-input state of the augmented model.

Thus, the reported implementation uses the single-half-space analytical projection of \eqref{eq:closed_form_solution} followed by the exact finite-rate propagation in \eqref{eq:finite_rate_command_update} and physical command clipping in \eqref{eq:finite_rate_command_clip}.
The same command limits are included in the composite safety certificate.
In the finite-rate implementation, a final component-wise clipping operation also enforces these limits on the executed command.
The continuous-time guarantee of Theorem~\ref{thm:continuous_composite_safety} therefore remains scoped to the continuous augmented model.

\subsection{Guarantee Scope}
\label{app:guarantee_scope}
Theorem~\ref{thm:continuous_composite_safety} is a continuous-time result for the reduced safety model.
Two sampling gaps remain outside that theorem: (i) the backup trajectory is checked only at the selected $\tau_k$, and (ii) the safety filter is recomputed only at the $10\,\mathrm{Hz}$ high-level updates.
In the reported finite-rate implementation, the same physical command limits included in the composite certificate are also enforced by a final
component-wise clipping operation.
This ensures that the executed command remains within its prescribed bounds, but the continuous-time theorem does not guarantee that the clipped command satisfies the composite CBF condition between high-level updates.

The LiDAR certificate is also local to the current scan. During one filter
computation, observed endpoints are fixed while the reduced pose and command
states evolve relative to them. The result therefore assumes that this local
model and the zero-surrogate backup adequately represent the high-level
motion over the prediction horizon. It does not certify full-body tracking
error, unobserved or occluded geometry, moving obstacles absent from the
model, or arbitrary LiDAR noise. Consequently, the formal guarantee applies
to the stated reduced safety model and measured scan-defined obstacles, not
to every full quadruped state.

\begin{assumption}[Finite-rate experimental realization]
\label{ass:finite_rate_implementation}
Over each high-level interval, the reduced pose and first-order command
dynamics are assumed to remain representative of the high-level closed-loop
motion, and the frozen $50\,\mathrm{Hz}$ locomotion controller is assumed to
track the resulting high-level command sufficiently well for that reduced
model to remain applicable.
\end{assumption}

\section{SAFTER--SAC Training Details}
\label{app:sac_details}

\subsection{Critic and Actor Updates}
\label{app:critic_actor_update}
The actor uses a reparameterized tanh-Gaussian policy in normalized action
coordinates $\tilde u_t^{\mathrm{nom}}\in[-1,1]^m$. The corresponding physical
nominal command is
$u_t^{\mathrm{nom}}=c_u+S_u\tilde u_t^{\mathrm{nom}}$ as defined in
Appendix~\ref{app:pomdp_details}. The complete finite-rate safety map consists
of the nominal-surrogate construction, the analytical half-space projection
in \eqref{eq:closed_form_solution}, the finite-rate command update in
\eqref{eq:finite_rate_command_update}, and the final component-wise
clipping in \eqref{eq:finite_rate_command_clip}. It produces
$u_t^{\mathrm{safe}}=\mathcal F_\eta(\xi_t,u_t^{\mathrm{nom}})$, which is
mapped back to normalized critic coordinates as
$\tilde u_t^{\mathrm{safe}}=S_u^{-1}(u_t^{\mathrm{safe}}-c_u)$.

For replay sample
$(o_t,\tilde u_t^{\mathrm{safe}},r_t,o_{t+1},\chi_t)$, where
$\chi_t$ is the true-termination indicator, the next nominal action is
filtered with the same forward map without gradient propagation to obtain
$\tilde u_{t+1}^{\mathrm{safe}}$. With twin critics
$Q_{\vartheta_1},Q_{\vartheta_2}$ and target networks
$Q_{\bar\vartheta_1},Q_{\bar\vartheta_2}$, the SAC target is
\[
    \mathcal Y_t
    =r_t+\gamma(1-\chi_t)
    \left[
        \min_{i=1,2}Q_{\bar\vartheta_i}
        (o_{t+1},\tilde u_{t+1}^{\mathrm{safe}})
        -\alpha_{\mathrm{ent}}
        \log\pi_\phi(\tilde u_{t+1}^{\mathrm{nom}}\mid o_{t+1})
    \right],
\]
and the critics minimize
\[
    J_Q
    =\mathbb E_{\mathcal D}\!\left[
       \sum_{i=1}^{2}
       \bigl(Q_{\vartheta_i}(o_t,\tilde u_t^{\mathrm{safe}})-\mathcal Y_t\bigr)^2
    \right].
\]
Thus the critics evaluate the action that actually generated the transition.

For the actor update, the critic parameters are fixed and
\[
    J_\pi(\phi)
    =\mathbb E\!\left[
        \alpha_{\mathrm{ent}}
        \log\pi_\phi(\tilde u_t^{\mathrm{nom}}\mid o_t)
        -\min_{i=1,2}Q_{\vartheta_i}(o_t,\tilde u_t^{\mathrm{safe}})
    \right].
\]
Gradients propagate through $\tilde u_t^{\mathrm{nom}}
    \mapsto u_t^{\mathrm{nom}}
    \mapsto \hat u_t^{\mathrm{nom}}
    \mapsto \hat u_t^\star
    \mapsto u_t^{\mathrm{pre}}
    \mapsto u_t^{\mathrm{safe}}
    \mapsto \tilde u_t^{\mathrm{safe}},$
where $\hat u_t^\star$ is the analytical half-space projection and
$u_t^{\mathrm{pre}}$ is the finite-rate command before clipping.

Holding the current filter state and command history fixed, define the local
projection Jacobian
\[
J_{\mathrm{proj},t}
:=
\frac{\partial \hat u_t^\star}
     {\partial \hat u_t^{\mathrm{nom}}}.
\]
Away from the switching boundary,
\[
J_{\mathrm{proj},t}
=
\begin{cases}
I_m, & \Omega_t \ge 0,\\[1mm]
I_m-\dfrac{a_t^\top a_t}{a_ta_t^\top}, & \Omega_t<0.
\end{cases}
\]
Thus, the analytical projection removes the component of the actor gradient
normal to the active safety half-space.

For the implemented finite-rate map, the corresponding local physical-action
Jacobian is
\begin{equation}
J_{F,t}
=
D_{\mathrm{clip},t}
B_\Delta
J_{\mathrm{proj},t}
A^{-1}
\left(
\frac{\iota_{t-1}}{\Delta t}I_m+\Sigma
\right),
\label{eq:full_filter_jacobian}
\end{equation}
where $B_\Delta=I-e^{-A\Delta t}$ and $D_{\mathrm{clip},t}$ denotes the local derivative of the component-wise clipping operation.
The surrounding mappings between normalized and physical action coordinates are handled by automatic differentiation.

The Identity-Jacobian ablation preserves the same numerical forward action but replaces only $J_{\mathrm{proj},t}$ by $I_m$ during the actor backward pass.
All other mappings in the safety-filter chain retain their ordinary
derivatives.

\subsection{Entropy Temperature and Critic Parameterization}
\label{app:entropy_temperature}
Entropy regularization remains defined on the nominal stochastic policy,
because its tanh-Gaussian log-density is tractable whereas the deterministic
filtered action does not generally have a tractable induced density. The
implementation uses automatic entropy tuning with
$\vartheta_\alpha:=\log\alpha_{\mathrm{ent}}$ and
\[
    J_\alpha(\vartheta_\alpha)
    =\mathbb E\!\left[
      -\vartheta_\alpha
      \bigl(\log\pi_\phi(\tilde u_t^{\mathrm{nom}}\mid o_t)
      +\mathcal H_{\mathrm{target}}\bigr)
    \right],
\]
using $\mathcal H_{\mathrm{target}}=-3$ for the three-dimensional command.
The log-probability is detached in this temperature update.

The critics are parameterized as $Q_\vartheta(o_t,\tilde u_t^{\mathrm{safe}})$.
Because the correction reward depends on both the nominal and projected surrogate commands, distinct nominal actions can produce the same executed action while incurring different correction costs.
Moreover, nominal-command history is included in the subsequent observation.
The critic can therefore alias such nominal-action choices and should be viewed as a practical observation-conditioned approximation rather than an exact Markov action-value representation.
Accordingly, ``exact'' in the Jacobian discussion refers to the local
analytical projection derivative, not to an exact policy gradient for the complete objective.

\subsection{Time-Limit Truncation}
\label{app:time_limit_truncation}
The Bellman mask uses $\chi_t=1$ for true task termination and $\chi_t=0$ for
time-limit truncation. For a truncation, replay stores the final observation
before reset as $o_{t+1}$, so the SAC target continues to bootstrap through
the artificial time limit.

\subsection{High-Level Reward Function}
\label{app:reward}

At each $10\,\mathrm{Hz}$ policy step, the reward is the
weighted sum of the components in
Table~\ref{tab:reward_weights}.
All SAC-based methods in our framework use the same task
terms, with the correction term disabled for Vanilla SAC and the correction-reward-off ablations.

\begin{table}[ht]
\centering
\caption{High-level reward components and signed weights.
The correction weight is $-0.05$ by default and zero for
Vanilla SAC and correction-reward-off ablations.}
\label{tab:reward_weights}
\small
\setlength{\tabcolsep}{5pt}
\renewcommand{\arraystretch}{1.3}
\begin{tabular}{@{}llr@{}}
\toprule
Component & Definition & Weight \\
\midrule
Goal progress
& $r_t^{\mathrm{prog}}=d_{t-1}-d_t$
& $+4.00$ \\
Goal reached
& $r_t^{\mathrm{goal}}=\mathbb I[d_t\leq d_{\mathrm{goal}}]$
& $+30.00$ \\
Time
& $r_t^{\mathrm{time}}=1$
& $-0.01$ \\
Timeout
& $r_t^{\mathrm{timeout}}=\mathbb I[\mathcal T_t]$
& $-5.00$ \\
Forward alignment
& $r_t^{\mathrm{align}}
=\dfrac{\operatorname{ReLU}(v_{x,t}^{\mathrm{meas}})}
{v_{x,\max}}\cos\beta_t$
& $+0.20$ \\
Reverse motion
& $r_t^{\mathrm{rev}}
=\left(\dfrac{\operatorname{ReLU}(-v_{x,t}^{\mathrm{meas}})}
{v_{x,\max}}\right)^2$
& $-0.50$ \\
Yaw tracking
& $r_t^{\mathrm{yaw}}
=\left(\dfrac{\omega_{z,t}^{\mathrm{meas}}
-\omega_t^{\mathrm{ref}}}{\omega_{z,\max}}\right)^2$
& $-0.05$ \\
Command smoothness
& $r_t^{\mathrm{smooth}}
=\|u_t^{\mathrm{safe}}-u_{t-1}^{\mathrm{safe}}\|_2^2$
& $-0.01$ \\
Filter correction
& $r_t^{\mathrm{corr}}
=\|(\hat u_t^{\mathrm{filt}}-\hat u_t^{\mathrm{nom}})
\oslash s_u\|_2^2$
& $-0.05/0$ \\
Unsafe termination
& $r_t^{\mathrm{coll}}=\mathbb I[\mathcal C_t]$
& $-20.00$ \\
\bottomrule
\end{tabular}
\end{table}

Here, $d_t$ is the Euclidean distance to the goal,
$d_{\mathrm{goal}}=0.35\,\mathrm m$, and $\beta_t$
is the signed body-frame goal bearing.
The measured forward velocity and yaw rate are
$v_{x,t}^{\mathrm{meas}}$ and
$\omega_{z,t}^{\mathrm{meas}}$.
Their command limits are $v_{x,\max}=2\,\mathrm{m/s}$
and $\omega_{z,\max}=1\,\mathrm{rad/s}$.
The reference yaw rate is
$\omega_t^{\mathrm{ref}}
=\operatorname{clip}(k_\beta\beta_t,
-\omega_{z,\max},\omega_{z,\max})$,
with $k_\beta=1.5\,\mathrm{s}^{-1}$.
The indicator $\mathbb I[\cdot]$ equals one when its
condition holds, and $\operatorname{ReLU}(q)=\max(0,q)$.
The event $\mathcal C_t$ denotes termination due to
obstacle collision, boundary-wall contact, or falling.
The event $\mathcal T_t$ denotes time-limit truncation
without prior goal or unsafe termination.
For SAFTER, we used $\hat u_t^{\mathrm{filt}}=\hat u_t^\star$, while for other filtered variants it denotes the corresponding filter output.
The positive vector $s_u$ contains componentwise
correction-normalization scales, and $\oslash$ denotes
elementwise division.

\subsection{SAC Network Architecture and Optimization}
\label{app:sac_implementation}
All SAC-based methods in our framework share hidden layers of 512, 256, and 128 units with ELU activations.
The actor is a reparameterized tanh-squashed diagonal Gaussian policy outputting three means and three log-standard deviations,
the latter clamped to $[-5,2]$.
Twin Q-networks use separate Polyak-averaged targets with $\tau_{\mathrm{P}}=0.005$.
The actor, critics, and entropy temperature use Adam with learning rate $3\times10^{-4}$.
Training uses discount $\gamma=0.99$, target entropy $-3$, a replay capacity of $2\times10^6$ transitions, and minibatches
of 1024.
Updates begin after 20,000 environment transitions, with one update per environment step.


\end{document}